\documentclass[11pt,a4paper]{article}

\usepackage[margin=2.5cm]{geometry}
\usepackage{amsmath,amssymb,amsthm}
\usepackage{enumitem}
\usepackage{booktabs}
\usepackage{array}
\usepackage{hyperref}
\hypersetup{
  hidelinks,
  pdftitle={Relative Prime Factorization and Finite-State Presentations under Fixed Finite-Monoid Observation},
  pdfauthor={Takayuki Kuriyama}
}
\usepackage{xurl}

\newtheorem{theorem}{Theorem}[section]
\newtheorem{lemma}[theorem]{Lemma}
\newtheorem{proposition}[theorem]{Proposition}
\newtheorem{corollary}[theorem]{Corollary}
\newtheorem{definition}[theorem]{Definition}
\newtheorem{example}[theorem]{Example}
\newtheorem{remark}[theorem]{Remark}
\newtheorem{question}[theorem]{Question}

\newcommand{\eps}{\varepsilon}
\newcommand{\D}{\mathcal D}
\newcommand{\Pcal}{\mathcal P}
\newcommand{\Vcal}{\mathcal V}
\newcommand{\Qcal}{\mathcal Q}
\newcommand{\Fact}{\operatorname{Fact}}
\newcommand{\Spec}{\operatorname{Spec}}
\newcommand{\Val}{\operatorname{Val}}
\newcommand{\Corr}{\operatorname{Corr}}
\newcommand{\Pleon}{\operatorname{Pleon}}
\newcommand{\fst}{\operatorname{fst}}
\newcommand{\lst}{\operatorname{lst}}

\newcommand{\Acc}{\operatorname{Acc}}
\newcommand{\Nat}{\mathbb N}
\newcommand{\Sat}{\operatorname{Sat}}
\newcommand{\Unsat}{\operatorname{Unsat}}
\newcommand{\Ex}{\operatorname{Ex}}

\newcommand{\CS}{\mathcal{C}\!\mathcal{S}}
\newcommand{\Can}{\operatorname{Can}}

\title{Relative Prime Factorization and Finite-State Presentations\\under Fixed Finite-Monoid Observation}
\author{Takayuki Kuriyama\\
\small Independent Researcher, Tokyo, Japan\\
\small \texttt{growup.kuriyama@gmail.com}}
\date{September 2, 2026}

\begin{document}
\maketitle

\begin{abstract}
Let $L\subseteq\Sigma^*$ and fix a morphism $h:\Sigma^*\to M$ into a finite monoid.  We study exact factorization and canonical presentation in the relative syntactic congruence $\theta_{L,h}:=\equiv_L\cap\ker h$.

Our main result separates unique factorization from finite direct presentation.  We construct an explicit, exhaustively computer-checked $36$-element quotient in which every live non-unit relative class has a unique exact prime factorization, while the valid prime-return rules contain the infinite family
\[
  [ab]\to[a][baaa]^m[b]\qquad(m\ge0).
\]
Hence unique factorization does not imply the finite relative presentation property (FRP), even for a finite quotient.  The witness uses the trivial language $\Sigma^+$, showing that the obstruction is already a finite-monoid phenomenon.  We then lift the same defect into the nonregular context-free language
\[
  L^{\mathrm{wrap}}=\{c^nwd^n:n\ge0,\ w\in\{a,b\}^+\}.
\]
For a finite observer extending the $36$-element one, this pair has an infinite relative quotient, exactly $52$ relative primes, global exact unique factorization, FSRP, and failure of FRP.  Thus the finite-state separation is not an artifact of a finite quotient or of a disjoint-alphabet embedding.

We isolate the obstruction as tail under-saturation and introduce the finite-state relative presentation property (FSRP), in which the canonical valid right-hand-side languages are represented by their finite residual controllers.  We prove $\mathrm{FRP}\subsetneq\mathrm{FSRP}$ and give a structural decomposition of direct-rule infinitude.  For $h$-substitutable context-free languages with finite relative prime spectrum, every correct prime-return language is context-free; the unresolved step is whether factor-minimalization can force a nonregular valid-return language.  In particular, we are not aware of any finite-prime $\neg$FSRP example.

On a rigid branch, prime-target left-division determinism (PTLD) restores Clark-style cancellation phenomena and implies unique exact factorization, tail exactness, tail determinism, and a quadratic valid-rule bound.  A five-prime nonregular deterministic context-free example with a finite group observer satisfies PTLD while lying outside every fixed $k,\ell$-substitutable class of Yoshinaka.  Finally, for fixed $h$ we give a strong positive-data learner for the canonical PTLD presentation with polynomial-time hypothesis updates and an explicit finite characteristic sample, and a more general limit reconstruction of the canonical FSRP controller from any weakly behaviorally correct CFG-valued positive-data learner.
\end{abstract}

\section{Introduction}

Distributional learning of context-free languages is based on the idea that substrings that occur in the same environments can be treated alike.  Clark and Eyraud formalized this through syntactic distributions and substitutability, while Yoshinaka refined the comparison relation by bounded left and right contexts~\cite{ClarkEyraud2007,Yoshinaka2008}.  The fixed finite-monoid framework replaces bounded windows by a specified finite algebraic observation $h:\Sigma^*\to M$: two strings are compared only when they have the same $h$-type, and observed overlap then forces equality of full two-sided distributions.

The present paper studies a different but closely related question.  Once the relative syntactic congruence
\(\theta_{L,h}=\equiv_L\cap\ker h\)
has been fixed, what intrinsic multiplicative structure do its classes possess, and when does that structure admit a finite canonical grammar presentation?

Clark's strong-learning construction for substitutable languages introduced prime syntactic congruence classes and valid productions.  In that setting, every nonzero nonunit class has a unique prime factorization and only finitely many valid productions occur when the prime spectrum is finite~\cite{Clark2013Trees}.  The relative setting is more delicate.  The finite observation $h$ refines syntactic classes and destroys some of the cancellation phenomena used in Clark's proof.  Exact factorization may become nonunique, and even a finite relative quotient need not imply a finite direct valid-rule basis.

The answer begins with a separation.  Relative unique prime factorization (UF) need not produce a finite direct relative-prime presentation (FRP), even when the relative quotient itself is finite.  The failure is not factorization ambiguity: in our $36$-element witness every live class has exactly one exact prime factorization, while valid prime-return paths have unbounded under-saturated tails.  A useful feature of the separation is that it is already visible at the finite-monoid level.  The target language of the flagship witness is deliberately chosen to be the trivial language $\Sigma^+$.  Consequently its nonempty syntactic structure contributes no complexity: on nonempty words the relative congruence is simply the kernel of the observer.  Thus the failure is caused by the gap between exact fiber products and multiplication in a finite observer quotient, not by a complicated context-free target language.  A matching-wrapper construction later shows that the same under-saturation mechanism persists intrinsically inside a nonregular context-free language with infinite relative quotient: the wrapper contributes the unbounded syntactic geometry, while the middle $36$-element dynamics still forces infinitely many valid prime returns.  Finite-state relative presentation (FSRP) preserves these infinite structural families through finite residual controllers.  Thus the central line of the paper is
\[
  \mathrm{UF}\not\Rightarrow\mathrm{FRP}
  \quad\longrightarrow\quad
  \text{under-saturation}
  \quad\longrightarrow\quad
  \mathrm{FSRP}.
\]

Three levels must therefore be distinguished:
\begin{enumerate}[label=(\roman*)]
\item \emph{exact prime factorization}, where ``exact'' means equality of sets of strings, not merely equality in a quotient;
\item \emph{correct and valid prime returns}, which concern only multiplication in the relative quotient;
\item \emph{finite-state presentability}, which allows an infinite direct rule family to be represented by a canonical finite-state controller.
\end{enumerate}
This separation is also consonant with the syntactic-concept viewpoint, where CFG nonterminals and productions are interpreted through sets of strings and inclusion/product structure rather than only individual quotient elements~\cite{Clark2015SCL}.

\paragraph{Main contributions.}
The four principal results are as follows.
\begin{enumerate}[leftmargin=2.2em]
\item An explicit, exhaustively computer-checked $36$-element quotient has fifteen relative primes and unique exact factorization for all $35$ live non-unit classes, yet it admits the infinite valid family $[ab]\to[a][baaa]^m[b]$.  Hence $\mathrm{UF}\not\Rightarrow\mathrm{FRP}$, and a separate finite FRP/non-UF witness shows that UF and FRP are incomparable.
\item FRP is strictly contained in FSRP.  Beyond the finite $36$-element witness, the nonregular context-free language $L^{\mathrm{wrap}}=\{c^nwd^n:n\ge0,\ w\in\{a,b\}^+\}$ has infinite relative quotient, exactly $52$ relative primes, global exact UF, FSRP, and not FRP.  Thus finite-state compression is genuinely needed even when the ambient relative structure is infinite.  It remains open whether an $h$-substitutable context-free pair with finite prime spectrum can fail FSRP itself.
\item The saturation-defect decomposition isolates the source of direct-rule infinitude.  Saturated valid tails have a finite quantitative bound; every failure of FRP localizes to an infinite family of under-saturated tails in one semantic residual class.
\item PTLD---uniqueness of right division into prime observer targets---restores Clark-style path rigidity.  It implies UF, valid-tail exactness, valid-tail determinism, and the $q^2$ direct-rule bound, while remaining strictly weaker than factor cancellation.
\end{enumerate}

Two consequences complete the picture.  First, $L_0=\{a^nb^n:n\ge0\}^*$ shows that PTLD is not necessary, whereas a five-prime nonregular deterministic context-free language with a $C_2\times C_2$ observer satisfies PTLD but is not $k,\ell$-substitutable for any finite $k,\ell$.  Thus fixed finite-monoid typing goes strictly beyond the entire bounded-context hierarchy, not merely beyond ordinary substitutability.  Second, for each fixed $h$ the PTLD branch admits strong reconstruction of the direct canonical grammar with polynomial-time hypothesis updates and an explicit finite characteristic sample.  On the lexically anchored PTLD subbranch the cut-separation radius vanishes, yielding polynomial characteristic data in explicit canonical grammar size; the same five-prime separation witness belongs to this subbranch.  Arbitrary compact CFG size cannot control canonical thickness polynomially, so transfer to external grammar representations requires additional restrictions.  The FSRP branch admits computable strong reconstruction of the canonical residual controller from any weakly behaviorally correct CFG-valued positive-data learner.  Prime existence, finite residual splitting, and minimality of the valid basis provide the structural infrastructure for these results rather than separate headline claims.

\section{Related work and positioning}
\label{sec:related}

The point of departure is Clark's canonical prime grammar for substitutable context-free languages, where syntactic congruence classes are factored into primes and valid productions provide a canonical strong-learning target~\cite{Clark2013Trees}.  Yoshinaka's $k,\ell$-substitutability weakens the ordinary condition by protecting the substituted middle with fixed boundary words $u\in\Sigma^k$ and $v\in\Sigma^\ell$: if two strings $uy_1v$ and $uy_2v$ occur in one common outer context, then they must remain interchangeable in every outer context~\cite{Yoshinaka2008}.  The case $k=\ell=0$ is ordinary substitutability.  Clark's later SCL-based constructions pursue a different finite-presentation principle: semantic prime sequences are ordered by inclusion, and finite canonical grammars are obtained by selecting maximal sequences under Noetherian-type hypotheses~\cite{Clark2015Canonical}.  In Clark's ordinary substitutable setting, unique factorization, tail rigidity, and finite valid-rule presentation occur together; finite-monoid refinement separates these phenomena.  Our FSRP construction keeps the structurally valid prime sequences themselves and, when they are infinite, compresses the family by residual languages rather than replacing it with semantic maxima.

\begin{center}
\scriptsize
\setlength{\tabcolsep}{3pt}
\begin{tabular}{>{\raggedright\arraybackslash}p{0.13\textwidth}>{\raggedright\arraybackslash}p{0.16\textwidth}>{\raggedright\arraybackslash}p{0.16\textwidth}>{\raggedright\arraybackslash}p{0.16\textwidth}>{\raggedright\arraybackslash}p{0.22\textwidth}}
\toprule
 & Yoshinaka 2008 & Clark 2013 & Clark 2015 & This paper\\
\midrule
substitution guard & fixed $k,\ell$ boundary words & ordinary shared-context condition & semantic SCL structure & fixed finite-monoid type $h$\\
semantic carrier & guarded substrings & syntactic classes & SCL closed sets & relative syntactic classes\\
structural atoms & --- & primes & SCL primes and irreducibles & relative primes\\
finiteness mechanism & bounded-context substitutability & substitutable rigidity & Noetherian maxima & FRP or residual compression\\
strong target & language identification & prime CFG & SCL grammar & direct CFG or residual controller\\
\bottomrule
\end{tabular}
\end{center}

The bounded-context guard and the fixed-monoid guard are genuinely different restrictions.  Section~\ref{sec:strong-learning} gives a single nonregular deterministic context-free language $L^\star$ that is $h$-substitutable and PTLD for a four-element group observer, yet is not $k,\ell$-substitutable for any finite pair $(k,\ell)$.  Thus the fixed-$h$ PTLD branch is not contained in the union of Yoshinaka's finite-window classes.

The present fixed-$h$ condition specializes the author's earlier framework of relation substitutability~\cite{KuriyamaFixedH}.  The first version of that arXiv record introduced the recognizable-relation condition in 2014; its current version develops the fixed finite-monoid weak-learning framework.  Taking the auxiliary recognizable equivalence to be $\ker h$ gives the same substitution-safety condition.  The extra role of the fixed morphism here is structural: quotient multiplication and observer-type division support exact relative factorization, PTLD, residual splitting, and the FRP/FSRP hierarchy.  The novelty claimed here is therefore not fixed-$h$ substitutability itself, but exact factorization in $\theta_{L,h}$, PTLD, the FRP/FSRP and under-saturation theory, and the resulting strong structural targets.

\paragraph{Extended CFGs and regular right-hand sides.}
Grammars in which the right-hand side of a production is described by a regular language are classical, appearing under such names as extended context-free grammars and regular-right-part grammars~\cite{MadsenKristensen1976,LaLonde1977,LaLonde1979}.  Accordingly, the fact that a regular family of right-hand sides can be compiled into an ordinary CFG is not new.  FSRP uses this classical expressive mechanism in a different role.  For each intrinsic relative prime $P$, the right-hand-side language $\Val_P$ is determined canonically by quotient return and prime-irreducibility.  FSRP is precisely the condition that these structurally determined languages are regular.  The residual controller is then their canonical minimal deterministic partial realization; adjoining the empty residual as a dead state recovers the usual complete realization.  The examples below establish that this finite-state level is strictly broader than a finite direct rule list, including for a nonregular context-free target with infinite relative quotient, but they do not establish that FSRP itself can fail under the context-free finite-prime hypotheses studied here.  Thus the contribution is not extended-CFG expressiveness itself, but the relative factorization theory that determines which right-hand-side languages must be represented, explains why a finite direct list of rules may fail, and isolates regularity of their factor-minimal part as a separate structural question.

\paragraph{Prime decomposition of languages and codes.}
Prime decomposition under language concatenation has an independent history in formal-language theory; see, for example, Han et al.~\cite{HanEtAl2007} for decomposition of regular languages.  Our notion is more constrained: the factors are not arbitrary languages but live classes of a fixed relative syntactic congruence, and exact setwise factorization is studied simultaneously with multiplication and return paths in the associated quotient.  The factor-free property of the global valid-rule language also places one aspect of the construction in classical code theory.  We use ``infix code'' only in this standard combinatorial sense; for broader background on codes and automata see Berstel, Perrin, and Reutenauer~\cite{BerstelPerrinReutenauer2009}.

\paragraph{Grammatical inference.}
The learning section uses the characteristic-sample viewpoint of polynomial grammatical inference introduced by de la Higuera~\cite{Higuera1997}.  It is also related to Clark's congruence-based CFG learner with a minimally adequate teacher~\cite{Clark2010MAT}, Yoshinaka's multidimensional substitutability for MCFGs~\cite{Yoshinaka2011}, and the subsequent distributional learning work of Clark and Yoshinaka~\cite{ClarkYoshinaka2014}.  The present strong-learning contribution differs in its target: the learner is required to recover the canonical relative-prime presentation determined by the fixed observer.

Residual finite-state automata already provide canonical automata whose states are residual languages~\cite{DenisLemayTerlutte2002}.  We do not claim residual automata as a new construction.  Here they are applied to the canonical languages $\Val_P$ of valid relative-prime right-hand sides, so that an infinite structural branching family is retained with its prime skeleton rather than replaced by an arbitrary DFA description of the terminal language.

Coste and Nicolas develop a reduction-based canonical form for local substitutable languages and show polynomial identification in time and thick data~\cite{CosteNicolas2019}.  Their reductions may compete and yield incomparable reduced alternatives.  Our valid-rule reduction has a related structural flavor, but the relative setting can admit infinitely many irreducible valid alternatives; FSRP replaces finite enumeration of those alternatives by finite-state residual compression.  Our prime-irreducibility means the absence of a prime-valued proper contiguous quotient block, rather than uniqueness of a reduction normal form in a parsing graph.

The residual-automaton construction and the Gold-style finite-automaton enumeration used later are standard ingredients.  The new claims concern the relative structural languages they represent, the UF--FRP separation, the under-saturation characterization, and the canonical structural-learning consequences.

\section{Relative syntactic classes}
\label{sec:relative}

Let $L\subseteq\Sigma^*$ be a language.  The two-sided distribution of $x\in\Sigma^*$ is
\(\D_L(x):=\{(u,v)\in\Sigma^*\times\Sigma^*:uxv\in L\}.\)
The syntactic congruence is
\(x\equiv_L y \iff \D_L(x)=\D_L(y).\)
Fix a monoid morphism $h:\Sigma^*\to M$.  Unless stated otherwise $M$ is finite.  For later use, write
\(\Fact(L):=\{x\in\Sigma^*:\D_L(x)\neq\varnothing\}\)
for the set of factors of $L$.

\begin{definition}[$h$-substitutability]
The language $L$ is \emph{$h$-substitutable} if, for all $x,y\in\Sigma^*$,
\(h(x)=h(y) \quad\text{and}\quad \D_L(x)\cap\D_L(y)\neq\varnothing\)
imply
\(x\equiv_L y.\)
Equivalently, among words of the same $h$-type, sharing one accepting two-sided context forces equality of the full syntactic distribution.
\end{definition}

This is the fixed-morphism specialization of relation-substitutability obtained by taking the auxiliary recognizable equivalence to be $\ker h$~\cite{KuriyamaFixedH}.

\begin{definition}[Relative syntactic congruence]
Define
\(\theta_{L,h}:=\equiv_L\cap\ker h.\)
Thus
\(x\equiv_{L,h}y \iff x\equiv_L y\ \text{and}\ h(x)=h(y).\)
Write $[x]_{L,h}$, or simply $[x]$, for the corresponding class.
\end{definition}

\begin{lemma}[Type--context collapse]
\label{lem:type-context-collapse}
If $L$ is $h$-substitutable, $h(x)=h(y)$, and
$\D_L(x)\cap\D_L(y)\ne\varnothing$, then $x\theta_{L,h}y$.
In particular, if $uxv,uyv\in L$ and $h(x)=h(y)$, then
$x\theta_{L,h}y$.
\end{lemma}

\begin{proof}
Immediate from the definitions of $h$-substitutability and $\theta_{L,h}$.
\end{proof}

Since both $\equiv_L$ and $\ker h$ are monoid congruences, so is $\theta_{L,h}$.

\begin{definition}[Live and dead classes]
A relative class $X$ is \emph{live} if $\D_L(x)\neq\varnothing$ for $x\in X$, and \emph{dead} otherwise.
\end{definition}

Unlike the ordinary syntactic congruence, the relative congruence may have several dead classes, because dead words with different $h$-values cannot be identified.  For each $m\in M$ there is at most one dead class of $h$-type $m$, hence the number of dead relative classes is at most $|M|$.

\begin{definition}[Unit separation]
We say that $(L,h)$ is \emph{unit-separated} if
\[
  [\eps]_{L,h}=\{\eps\}.
  \tag{US}
\]
\end{definition}

Unit separation is exactly what is needed if lexical terminals are to be represented by non-unit prime classes.  It should be distinguished from cancellativity of the observation monoid.

\subsection{Setwise product versus quotient product}

If $X,Y\subseteq\Sigma^*$ are relative classes, define their exact setwise product by
\(X\cdot Y:=\{xy:x\in X,\ y\in Y\}.\)
Since $\theta_{L,h}$ is a congruence, the quotient product
\(X*Y:=[xy]_{L,h}\)
is well defined for any $x\in X,y\in Y$.  Always $X\cdot Y\subseteq X*Y$, but equality need not hold.

This distinction is fundamental.  Relative primality is defined using exact setwise equality, whereas Clark-style correct productions only require the quotient product to return to a prime class.  In ordinary syntactic notation Clark likewise uses $[XY]$ for the ambient congruence class containing the set product $XY$, and correct productions are based on the ambient class rather than equality of the set product~\cite{Clark2013Trees}.

\section{Relative primes and exact decomposition}

\begin{definition}[Relative prime]
A live relative class $P$ is \emph{prime} if it is non-unit and whenever
\(P=X\cdot Y\)
for relative classes $X,Y$, one of $X,Y$ is the unit class.  A live non-unit class that is not prime is \emph{composite}.
\end{definition}

For a nonempty class $X$, write
\(\ell(X):=\min\{|w|:w\in X\}.\)

\begin{lemma}[Shortest-word geometry of exact products]
\label{lem:shortest-word-cut}
Suppose $X=X_1\cdots X_k$ is an exact product of nonempty relative
classes.  Then $\ell(X)=\sum_{i=1}^k\ell(X_i)$.  Moreover, if
$w\in X$ has $|w|=\ell(X)$, then
$w=w_1\cdots w_k$ for some $w_i\in X_i$ with
$|w_i|=\ell(X_i)$ for every $i$.  In particular, every exact prime
factorization of $X$ occurs among the finitely many prime-labelled cuts
of any shortest word of $X$.
\end{lemma}

\begin{proof}
Every product word has length at least $\sum_i\ell(X_i)$, while
concatenating shortest representatives attains that value.  Hence
$\ell(X)=\sum_i\ell(X_i)$.  Any factorization of a shortest $w\in X$
through the exact product must attain equality in each component,
which proves the second assertion.
\end{proof}

\begin{corollary}[Finite candidate principle]
\label{cor:finite-candidates}
Every exact prime factorization of a live class $X$ occurs as a prime-labelled cut of one fixed shortest word of $X$.  In particular, every such factorization has length at most $\ell(X)$, and there are at most $2^{\ell(X)-1}$ cut candidates when $\ell(X)\ge1$.
\end{corollary}

\begin{theorem}[Existence of exact relative prime decompositions]
\label{thm:prime-existence}
Every live non-unit relative class is an exact setwise product of one or more relative primes.
\end{theorem}

\begin{proof}
Induct on $\ell(X)$.  If $X$ is prime there is nothing to prove.  Otherwise $X=Y\cdot Z$ with non-unit relative classes $Y,Z$.  Since $X$ is live, so are $Y$ and $Z$: if $yz$ occurs as a factor of a word of $L$, then both $y$ and $z$ occur as factors.  By the preceding lemma,
\(\ell(X)=\ell(Y)+\ell(Z),\)
and non-unitness gives $1\le\ell(Y),\ell(Z)<\ell(X)$.  Apply the induction hypothesis to $Y$ and $Z$ and concatenate their prime decompositions.
\end{proof}

\begin{remark}
No substitutability, cancellation, or finiteness assumption is used in this existence theorem.  This parallels the existence half of Clark's prime-factorization lemma but is valid at the level of an arbitrary $L$-compatible monoid congruence.
\end{remark}

\section{Correct, pleonastic, and valid productions}
\label{sec:valid}

Let $\Pcal$ be the set of live relative primes.  For a nonempty prime sequence
\(\alpha=P_1\cdots P_k\in\Pcal^+\)
write
\[
  \overline\alpha:=P_1\cdot\ldots\cdot P_k\subseteq\Sigma^*,
  \qquad
  [\overline\alpha]:=P_1*\cdots*P_k.
\]
For the empty prime sequence we set
\[
  \overline\eps:=\{\eps\},
  \qquad
  [\overline\eps]:=[\eps]_{L,h}.
\]

\begin{definition}[Correct branching production]
For $k\ge2$, a rule
\(P\to P_1\cdots P_k\)
is \emph{correct} if $P\in\Pcal$ and
\([\overline\alpha]=P.\)
Equivalently, $\overline\alpha\subseteq P$.
\end{definition}

This is the precise relative analogue of Clark's correct branching production $[\bar\alpha]\to\alpha$~\cite{Clark2013Trees}.

\begin{definition}[Pleonastic and valid]
A prime sequence $\alpha=P_1\cdots P_k$ is \emph{pleonastic} if it has a proper contiguous block
\(\beta=P_i\cdots P_j, \qquad j-i+1\ge2,\)
whose quotient product $[\overline\beta]$ is prime.  A correct production is \emph{valid} if its right-hand side is not pleonastic.
\end{definition}

For each prime $P$, put
\(\Val_P:=\{\alpha\in\Pcal^{\ge2}:P\to\alpha\text{ is valid}\}\)
and \(\Val:=\bigcup_{P\in\Pcal}\Val_P.\)
We identify a production $P\to\alpha$ with the pair $(P,\alpha)$ and write
\(\Vcal:=\{(P,\alpha):P\in\Pcal,\ \alpha\in\Val_P\}.\)
Thus, when $\Pcal$ is finite, $|\Vcal|=\sum_{P\in\Pcal}|\Val_P|$.

\begin{lemma}[Valid-rule reduction]
Every correct production
\(P\to\alpha\)
is derivable using valid productions only:
\(P\Rightarrow^*\alpha.\)
\end{lemma}

\begin{proof}
Induct on $|\alpha|$.  If the rule is valid, apply it directly.  Otherwise write
\(\alpha=\gamma\beta\delta\)
where $\beta$ is a proper contiguous block of length at least two and
\(Q:=[\overline\beta]\)
is prime.  Since quotient multiplication is associative and $\overline\beta\subseteq Q$, both
\(P\to\gamma Q\delta, \qquad Q\to\beta\)
are correct, and both right-hand sides are shorter than $\alpha$.  Apply the induction hypothesis to both rules.
\end{proof}

The proof is the same compression principle as Clark's reduction of correct rules to valid rules, but does not use uniqueness of prime factorization.

\subsection{Structural minimality of valid rules}

\begin{definition}[Structurally complete correct prime system]
A set $R$ of correct productions over $\Pcal$ is \emph{structurally complete} if every correct production $P\to\alpha$ satisfies
\(P\Rightarrow_R^*\alpha.\)
\end{definition}

\begin{theorem}[Minimal valid-basis theorem]
\label{thm:minimal-valid-basis}
The set $\Vcal$ of all valid productions is the least structurally complete correct prime-production system.
\end{theorem}

\begin{proof}
The valid-rule reduction lemma shows that $\Vcal$ is structurally complete.  Conversely, let $R$ be structurally complete and let $P\to\alpha$ be valid.  Consider an $R$-derivation tree from $P$ to the prime word $\alpha$.  Any internal non-root prime node $Q$ spans a contiguous block $\beta$ of the prime frontier.  By induction down the derivation subtree rooted at $Q$, correctness of its productions implies $[\overline\beta]=Q$.  Prime symbols here are grammar nonterminals; lexical leaves are terminal symbols and are not counted as internal prime nodes.  If $|\beta|\ge2$, then $\beta$ is a proper prime-valued contiguous block of $\alpha$, contradicting validity.  Hence no nontrivial internal branching node can occur below the root.  The root must therefore use the rule $P\to\alpha$ directly, so $P\to\alpha\in R$.  Thus $\Vcal\subseteq R$ for every structurally complete correct system $R$.
\end{proof}

\begin{corollary}
The valid right-hand-side language $\Val$ is factor-free: no valid word is a proper contiguous factor of another valid word.
\end{corollary}

\begin{proof}
If a valid word properly contained another valid word as a contiguous factor, that factor would have prime quotient and length at least two, making the larger word pleonastic.
\end{proof}

In code-theoretic terminology, the global valid-rule language $\Val$ is therefore an \emph{infix code}~\cite{ItoJurgensenShyrThierrin1991}.

\section{Canonical generation from a finite direct basis}

We first formulate the construction for an arbitrary monoid congruence $\theta\subseteq\equiv_L$.  We write $\mathrm{FRP}(L,\theta)$ and $\rho(L,\theta)$ for the general congruence-level notions; when $\theta=\theta_{L,h}$ we abbreviate them by $\mathrm{FRP}(L,h)$ and $\rho(L,h)$.  The same convention will be used below for FSRP and its state complexity.

\begin{definition}[Accepting classes]
An $L$-compatible congruence class $X$ is \emph{accepting} if $X\subseteq L$.  Let
\(\Acc_\theta(L):=\{X:X\subseteq L\}.\)
\end{definition}

Because $\theta\subseteq\equiv_L$, each class is either entirely contained in $L$ or disjoint from $L$.

\begin{definition}[Finite relative presentation property]
The pair $(L,\theta)$ has the \emph{finite relative presentation property} (FRP) if
\(|\Pcal|<\infty \qquad\text{and}\qquad |\Vcal|<\infty.\)
\end{definition}

When $|\Pcal|<\infty$, define the \emph{validity radius}
\(\rho(L,\theta):=\sup\{|\alpha|:\alpha\in\Val_P\text{ for some }P\in\Pcal\}.\)
Since the prime alphabet is finite, $\mathrm{FRP}(L,\theta)\iff\rho(L,\theta)<\infty$.

\begin{theorem}[Finite direct relative-prime presentation]
\label{thm:direct-presentation}
Suppose
\begin{enumerate}[label=(\roman*)]
\item $\theta\subseteq\equiv_L$;
\item $[\eps]_\theta=\{\eps\}$;
\item $\Acc_\theta(L)$ is finite;
\item $(L,\theta)$ satisfies FRP.
\end{enumerate}
Then there is a finite CFG $G^{\mathrm{rp}}_{L,\theta}$ such that for every live prime $P$,
\(L(G^{\mathrm{rp}}_{L,\theta},P)=P,\)
and
\(L(G^{\mathrm{rp}}_{L,\theta})=L.\)
\end{theorem}

\begin{proof}
Use the live primes as nonterminals, together with a fresh start symbol.  Include all valid branching productions and each lexical rule $[a]_\theta\to a$ for every letter whose class is live.  Unit separation implies that a live letter class is non-unit, and the length-one argument shows it is prime.

For each accepting non-unit class $X$, include start rules for its exact prime decompositions.  If $\eps\in L$, also include the start rule
\(S\to\eps.\)
Corollary~\ref{cor:finite-candidates} shows that each accepting class has only finitely many exact prime decompositions, so the start-rule set is finite.

Soundness follows because every branching rule is correct.  For completeness, take $w=a_1\cdots a_n\in P$.  If $n=1$, then $P=[a_1]_\theta$ and the lexical rule
\[
  P\to a_1
\]
derives $w$ directly.  Assume now $n\ge2$.  Then
\(P\to[a_1]_\theta\cdots[a_n]_\theta\)
is a correct branching production.  By valid-rule reduction,
\[
  P\Rightarrow^*[a_1]_\theta\cdots[a_n]_\theta\Rightarrow^*w.
\]
Thus each prime nonterminal generates exactly its class.  The start rules then generate exactly the accepting classes, hence $L$.
\end{proof}

\subsection{The role of fixed-\texorpdfstring{$h$}{h} substitutability}

For the fixed observation $h:\Sigma^*\to M$, define
\(L_m:=L\cap h^{-1}(m).\)
If $L$ is $h$-substitutable and $x,y\in L_m$, then $x$ and $y$ have the common context $(\eps,\eps)$ and the same $h$-type, hence $x\equiv_L y$.  Therefore each nonempty $L_m$ is a single $\theta_{L,h}$-class.  Consequently
\(|\Acc_{\theta_{L,h}}(L)|\le |M|.\)
Thus the internal prime calculus above does not need substitutability; in the fixed-$h$ theorem, substitutability supplies a finite bound on the accepting relative classes.

\section{Semantic residuals, prime-target division, and Clark rigidity}
\label{sec:residual-ptld}

Clark's ordinary substitutable theory combines unique prime factorization with a stronger tail-exactness phenomenon.  Under finite-monoid refinement these effects separate, and the correct intermediate object is the semantic residual.  Throughout this section,
\[
  \theta:=\theta_{L,h}.
\]
Accordingly, all relative classes and all residual sets $\operatorname{Res}_\theta(A,P)$ below refer to the fixed-$h$ relative syntactic congruence.

For relative classes $A,P$, define
\(A\backslash P:=\{x\in\Sigma^*:Ax\subseteq P\}.\)
This set is $\theta$-saturated.  Let
\[
  \operatorname{Res}_\theta(A,P)
  :=\{R:R\text{ is a live relative class and }R\subseteq A\backslash P\}.
\]

\begin{theorem}[Finite residual splitting]
\label{thm:finite-residual-splitting}
Assume $L$ is $h$-substitutable and $M$ is finite, and let $A,P$ be live relative classes.  For every $m\in M$,
\((A\backslash P)\cap h^{-1}(m)\)
is empty or a single relative class.  Consequently
\(|\operatorname{Res}_\theta(A,P)|\le |M|.\)
More sharply, if $a=h(A)$ and $p=h(P)$, then
\(|\operatorname{Res}_\theta(A,P)| \le |\{m\in \mathsf T_{L,h}:am=p\}|,\)
where $\mathsf T_{L,h}=h(\Fact(L))$.
\end{theorem}

\begin{proof}
Take $x,y\in A\backslash P$ with $h(x)=h(y)$ and choose $a_0\in A$.
Then $a_0x,a_0y\in P$.  Liveness of $P$ supplies an outer context in
which both occur, so $x$ and $y$ share a live context after the common
left factor $a_0$.  Lemma~\ref{lem:type-context-collapse} gives
$x\theta y$.  Finally, every $R\in\operatorname{Res}_\theta(A,P)$ is live, so its observer type $m=h(R)$ belongs to $\mathsf T_{L,h}$.  Since $AR\subseteq P$, every such type also satisfies $h(A)m=h(P)$, which gives the sharper bound.
\end{proof}

\begin{definition}[Valid-tail exactness]
A valid production
\(N\to\gamma\alpha, \qquad \gamma,\alpha\ne\eps,\)
has an \emph{exact tail} if
\(\overline\alpha=[\overline\alpha]_\theta.\)
The pair $(L,h)$ has \emph{valid-tail exactness} (VTE) if every such suffix is exact.  We write VTE$_1$ for the weaker condition requiring exactness only after the first right-hand prime, i.e. for every valid $P\to A\alpha$.
\end{definition}

Only VTE$_1$ is needed for the direct-presentation and rule-recovery consequences below; full VTE is retained because PTLD yields the stronger Clark-style suffix statement.

\begin{definition}[Valid-tail determinism]
We say that $(L,h)$ has \emph{valid-tail determinism} (VTD) if, whenever
\(N\to A\alpha, \qquad N\to A\beta\)
are valid productions with the same left-hand prime $N$ and the same first right-hand prime $A$, then $\alpha=\beta$.
\end{definition}

\begin{definition}[Prime-target left-division determinism]
Let
\(\mathsf T_{L,h}:=h(\Fact(L)).\)  Since $\eps\in\Fact(L)$ whenever $L\ne\varnothing$, one has $1\in\mathsf T_{L,h}$.
We say that $(L,h)$ satisfies \emph{prime-target left-division determinism} (PTLD) if for every $q,r,r'\in \mathsf T_{L,h}$ and every relative prime $P$,
\(qr=h(P)=qr' \quad\Longrightarrow\quad r=r'.\)
\end{definition}

Factor-left-cancellation on $\mathsf T_{L,h}$ implies PTLD, but PTLD asks for cancellation only in equations whose target is a prime observer type.

\begin{remark}[Groups as a sufficient source of PTLD]
If $M$ is a finite group, then PTLD holds automatically: $qr=p=qr'$ implies $r=q^{-1}p=r'$.  Groups are thus a convenient source of examples, although they are not necessary; Example~\ref{ex:ptld-not-cancel} below uses a non-group observer.
\end{remark}

\begin{remark}[Action-category interpretation]
PTLD is naturally a thinness condition in the right action category $\mathcal R(M)$, whose objects are elements of $M$ and whose hom-set is $\mathcal R(M)(q,p)=\{r\in M:qr=p\}$.  When the relevant factor types fill $M$, PTLD says exactly that every prime observer target $p=h(P)$ is subterminal: $|\mathcal R(M)(q,p)|\le1$ for every $q$.  This is stronger than triviality of the local stabilizer $\{r:pr=p\}$ because parallel arrows may also occur across strict Green $\mathcal R$-drops.  This one-sided action category should not be confused with the standard two-sided kernel category of semigroup theory~\cite{Tilson1987}; the viewpoint is interpretive rather than an additional hypothesis.
\end{remark}

\medskip
\noindent\textbf{Standing hypothesis (R).}\par\noindent
Within this section, \textup{(R)} abbreviates the assumptions that $L$ is
$h$-substitutable, $(L,h)$ is unit-separated, and PTLD holds.  Outside
this section we restate these three assumptions explicitly.

\begin{lemma}[Prime-target remainder collapse]
\label{lem:ptld-collapse}
Under \textup{(R)}, let $P$ be a relative prime, let $u\theta u'$, and suppose $ux,u'y\in P$.  Then $x\theta y$.
\end{lemma}

\begin{proof}
Because $P$ is live, all factors $u,u',x,y$ occurring in the displayed prime words are live factors; hence their observer types belong to $\mathsf T_{L,h}$.  Since $u\theta u'$, PTLD applied to
\[
  h(u)h(x)=h(P)=h(u')h(y)=h(u)h(y)
\]
gives $h(x)=h(y)$.  Choose $(\ell,r)$ with $\ell uxr\in L$.  Since $ux\theta u'y$, also $\ell u'yr\in L$, and since $u\equiv_Lu'$, also $\ell uyr\in L$.  Thus $x$ and $y$ share the context $(\ell u,r)$ and have the same observer type; Lemma~\ref{lem:type-context-collapse} gives $x\theta y$.
\end{proof}

\begin{lemma}[Prime left-quotient rigidity]
\label{lem:prime-left-quotient-rigidity}
Under \textup{(R)}, let $P$ be a relative prime and $u\in\Sigma^*$.  Then
\[
  u^{-1}P:=\{v:uv\in P\}
\]
is empty or exactly one relative class.  Moreover,
\[
  u^{-1}P=[\eps]_\theta
  \quad\Longleftrightarrow\quad
  u\in P.
\]
\end{lemma}

\begin{proof}
If $uv,uw\in P$, Lemma~\ref{lem:ptld-collapse} with the common prefix $u$ gives $v\theta w$.  Thus a nonempty left quotient is contained in one relative class.  It is also $\theta$-saturated: if $v\theta v'$ and $uv\in P$, then $uv\theta uv'$ and hence $uv'\in P$.  Therefore it is exactly that class.  Finally, $\eps\in u^{-1}P$ iff $u\in P$; when this holds, the unique class containing $\eps$ is $[\eps]_\theta=\{\eps\}$ by unit separation.
\end{proof}

\begin{corollary}[Prime prefix-freeness]
\label{cor:prime-prefix-free}
Under \textup{(R)}, no word of a relative prime is a proper prefix of another word of the same prime.
\end{corollary}

\begin{proof}
If $x,xc\in P$, then $c\in x^{-1}P=[\eps]_\theta$ by Lemma~\ref{lem:prime-left-quotient-rigidity}.  Unit separation gives $c=\eps$.
\end{proof}

\begin{lemma}[Relative prime-prefix escape]
\label{lem:relative-prefix-escape}
Under \textup{(R)}, if $X$ is a relative prime and $Y$ is a distinct live non-unit relative class, then some $x\in X$ does not begin with an element of $Y$.
\end{lemma}

\begin{proof}
Assume contrariwise that every $x\in X$ has a factorization $x=yv$ with $y\in Y$.  Fix one such factorization $x_0=y_0v_0$.  For any other $x=yv$, we have $y\theta y_0$ and $yv,y_0v_0\in X$, so Lemma~\ref{lem:ptld-collapse} gives $v\theta v_0$.  Hence every remainder lies in one relative class $R$.  Congruence now gives the exact equality $X=Y\cdot R$.  If $R$ is non-unit this contradicts primality of $X$; if $R$ is the unit class, unit separation gives $R=\{\eps\}$ and hence $X=Y$, again a contradiction.
\end{proof}

\begin{lemma}[Relative prime-prefix interception]
\label{lem:prefix-interception}
Under \textup{(R)}, let
\(\alpha=A_1\cdots A_m\) and \(\beta=B_1\cdots B_n\)
be nonempty sequences of relative primes.  If
\(\overline\alpha\supseteq\overline\beta,\)
then there is $j$, $1\le j\le n$, such that
\[
  A_1\supseteq\overline{B_1\cdots B_j}.
\]
\end{lemma}

\begin{proof}
If $A_1=B_1$, take $j=1$.  Otherwise choose $b_1\in B_1$ by Lemma~\ref{lem:relative-prefix-escape} so that $b_1$ does not begin with an $A_1$-word.  Extend $b_1$ to a word of $\overline\beta$ and factor that word through $\alpha$.  Its first $A_1$-factor must extend strictly beyond $b_1$, so $R_1:=b_1^{-1}A_1$ is nonempty.  By Lemma~\ref{lem:prime-left-quotient-rigidity}, $R_1$ is a single relative class.

Inductively, suppose $b_i\in B_i$ have been chosen through position $t$ and $R_t:=(b_1\cdots b_t)^{-1}A_1$ is a nonempty relative class.  If $R_t=[\eps]_\theta$, then $b_1\cdots b_t\in A_1$, and congruence gives $\overline{B_1\cdots B_t}\subseteq A_1$, so take $j=t$.  If $t<n$ and $R_t=B_{t+1}$, then congruence similarly gives $\overline{B_1\cdots B_{t+1}}\subseteq A_1$, and take $j=t+1$.

Otherwise $t<n$ and $R_t$ is a live non-unit class distinct from $B_{t+1}$.  By Lemma~\ref{lem:relative-prefix-escape}, choose $b_{t+1}\in B_{t+1}$ that does not begin with an $R_t$-word.  Extend $b_1\cdots b_{t+1}$ to a word of $\overline\beta$ and factor it through $\alpha$.  The first $A_1$-factor cannot end before $b_1\cdots b_t$: it would be a proper prefix of an existing $A_1$-word extending that prefix, contradicting Corollary~\ref{cor:prime-prefix-free}.  It cannot end at $b_1\cdots b_t$, since then $R_t=[\eps]_\theta$.  Nor can it end inside $b_{t+1}$, because then it would have the form $b_1\cdots b_ts$ with $s\in R_t$ a prefix of $b_{t+1}$, contrary to the choice of $b_{t+1}$.  Thus it extends strictly beyond $b_1\cdots b_{t+1}$, so $R_{t+1}$ is again nonempty.

If the construction reached $t=n$ without one of the terminating cases, the complete word $b_1\cdots b_n\in\overline\beta$ would have to possess an $A_1$-factor extending beyond its end, impossible.  Hence a required $j$ exists.
\end{proof}

\begin{lemma}[Shortest-word left division]
\label{lem:shortest-left-division}
Under \textup{(R)}, let $X$ be a relative prime and let $x$ be a shortest word in $X$.  For every set $S\subseteq\Sigma^*$,
\[
  x^{-1}(X\cdot S)=S.
\]
\end{lemma}

\begin{proof}
The inclusion $S\subseteq x^{-1}(X\cdot S)$ is immediate.  Conversely, if $xv=x's$ with $x'\in X$ and $s\in S$, then $x$ and $x'$ are comparable prefixes of the same word.  Minimality of $x$ gives $|x'|\ge|x|$, while Corollary~\ref{cor:prime-prefix-free} forbids $x$ from being a proper prefix of $x'$.  Hence $x=x'$ and $v=s\in S$.
\end{proof}

\begin{corollary}[Prime-left inclusion cancellation]
\label{lem:prime-left-cancellation}
Under \textup{(R)}, if $X$ is prime and $X\overline\alpha\subseteq X\overline\beta$, then $\overline\alpha\subseteq\overline\beta$.
\end{corollary}

\begin{proof}
Apply $x^{-1}$ for a shortest $x\in X$ and use Lemma~\ref{lem:shortest-left-division}.
\end{proof}

\begin{corollary}[Prime-left equality cancellation]
\label{cor:prime-left-equality}
Under \textup{(R)}, $X\overline\alpha=X\overline\beta$ implies $\overline\alpha=\overline\beta$.
\end{corollary}

\begin{corollary}[Prime-left saturation]
\label{lem:prime-left-saturation}
Under \textup{(R)}, if $X$ is prime, $C$ is a live relative class, and $C=X\overline\alpha$ holds exactly, then $\overline\alpha=[\overline\alpha]_\theta$.
\end{corollary}

\begin{proof}
Let $x$ be a shortest word in $X$.  Lemma~\ref{lem:shortest-left-division} gives $\overline\alpha=x^{-1}C$.  Since $C$ is a $\theta$-class, it is $\theta$-saturated, and the left quotient of a $\theta$-saturated set by a fixed word is again $\theta$-saturated.  The nonempty set $\overline\alpha$ is contained in the single class $[\overline\alpha]_\theta$, so saturation forces equality.  This also covers $\alpha=\eps$, using $\overline\eps=\{\eps\}$ and unit separation.
\end{proof}

\begin{corollary}[Peeling]
\label{cor:peeling}
Under \textup{(R)}, let $X$ be prime and let $C$ be a live relative class.  If $C=X\overline\alpha$ holds exactly and $X\overline\beta\subseteq C$, then $\overline\alpha$ is an exact relative class and $\overline\beta\subseteq\overline\alpha$.
\end{corollary}

\begin{proof}
Exactness is Corollary~\ref{lem:prime-left-saturation}, and the inclusion is Corollary~\ref{lem:prime-left-cancellation}.
\end{proof}

\begin{definition}[Prime-irreducible sequence]
A prime sequence $\beta=B_1\cdots B_m$ is \emph{prime-irreducible} if
no contiguous block of length at least two has prime quotient product;
that is, for every $1\le i<j\le m$,
$[\overline{B_i\cdots B_j}]_\theta$ is not a relative prime.
\end{definition}

\begin{lemma}[Prime contraction]
\label{lem:prime-contraction}
Let $\beta$ be a prime sequence and replace a contiguous block of length at least two whose quotient is a prime $Q$ by $Q$, obtaining $\beta'$.  Then $\overline\beta\subseteq\overline{\beta'}\subseteq[\overline\beta]_\theta$ and $[\overline{\beta'}]_\theta=[\overline\beta]_\theta$.  Hence iterated contractions terminate in a prime-irreducible sequence with the same quotient class; if $\overline\beta$ is already that whole class, every stage remains exact.
\end{lemma}

\begin{proof}
Replacing the block enlarges its setwise product to its quotient prime without changing quotient multiplication.  Each contraction strictly decreases sequence length, so iteration terminates.
\end{proof}

\begin{theorem}[Irreducible-path exactness]
\label{thm:irreducible-path}
Assume \textup{(R)}.  Let \(C=\overline{A_1\cdots A_k}\) be any exact prime
factorization of a live relative class $C$.  If
$\overline{B_1\cdots B_m}\subseteq C$ and $B_1\cdots B_m$ is
prime-irreducible, then $m=k$ and $B_i=A_i$ for every $i$.  In
particular, $\overline{B_1\cdots B_m}=C$.
\end{theorem}

\begin{proof}
Induct on $k$.  If $k=1$, then $C=A_1$ is prime.  If $m\ge2$, the
whole sequence $B_1\cdots B_m$ has prime quotient $C$, contradicting
prime-irreducibility.  Hence $m=1$, and the nonempty inclusion
$B_1\subseteq A_1$ gives $B_1=A_1$.

Now let $k\ge2$.  By Lemma~\ref{lem:prefix-interception}, some $j\ge1$
satisfies $\overline{B_1\cdots B_j}\subseteq A_1$.  If $j\ge2$, that
block has prime quotient $A_1$, contradicting prime-irreducibility.
Hence $j=1$ and $B_1=A_1$.  The case $m=1$ is impossible, because then
the prime class $B_1=A_1$ would be contained in the relative class
$C$, forcing $C=A_1$, contrary to the nontrivial exact factorization of
$C$.  Thus both tails are nonempty.  Corollary~\ref{cor:peeling} gives
\(\overline{B_2\cdots B_m}\subseteq\overline{A_2\cdots A_k}\), and
the right-hand side is an exact relative class.  The induction
hypothesis applies to the two tails.
\end{proof}

\begin{corollary}[Exact prime factorizations are irreducible]
\label{cor:exact-irreducible}
Under \textup{(R)}, every exact prime factorization is
prime-irreducible.
\end{corollary}

\begin{proof}
If an exact factorization were reducible, Lemma~\ref{lem:prime-contraction} would produce a strictly shorter prime-irreducible exact factorization of the same class.  Theorem~\ref{thm:irreducible-path}, applied to the original exact factorization and that terminal sequence, would force them to be identical, a contradiction.
\end{proof}

\begin{theorem}[Prime-target division factorization theorem]
\label{thm:ptld-uf}
Assume that $L$ is $h$-substitutable, unit-separated, and satisfies
PTLD.  Then every live non-unit relative class has a unique exact
factorization into relative primes.
\end{theorem}

\begin{proof}
Existence is Theorem~\ref{thm:prime-existence}.  If
$C=\overline\alpha=\overline\beta$ are two exact prime factorizations,
Corollary~\ref{cor:exact-irreducible} makes $\beta$ prime-irreducible.
Theorem~\ref{thm:irreducible-path}, applied to the exact factorization
$\alpha$ and the contained sequence $\beta$, gives $\alpha=\beta$.
\end{proof}

\begin{theorem}[PTLD implies tail exactness and valid-tail determinism]
\label{thm:ptld-vte-vtd}
Assume that $L$ is $h$-substitutable, unit-separated, and satisfies
PTLD.  Then VTE holds, and hence VTE$_1$ holds.  If the relative prime
spectrum has size $q<\infty$, then VTD holds and $|\Vcal|\le q^2$.
\end{theorem}

\begin{proof}
Let $N\to\gamma\alpha$ be valid, with $\gamma,\alpha\ne\eps$.  Every
contiguous block of $\alpha$ of length at least two is a proper
contiguous block of the full right-hand side $\gamma\alpha$.  Hence
validity makes $\alpha$ prime-irreducible.  Put
$T^{\mathrm{tail}}:=[\overline\alpha]_\theta$.  This class is live and
non-unit, and Theorem~\ref{thm:prime-existence} gives it an exact prime
factorization.  Since $\overline\alpha\subseteq T^{\mathrm{tail}}$,
Theorem~\ref{thm:irreducible-path} shows that $\alpha$ is that exact
factorization.  Thus $\overline\alpha=T^{\mathrm{tail}}$, proving VTE.

For VTD, suppose $N\to A\alpha$ and $N\to A\beta$ are valid.  VTE makes $R_\alpha:=\overline\alpha$ and $R_\beta:=\overline\beta$ exact live classes.  For any $a\in A$, $r\in R_\alpha$, and $r'\in R_\beta$, correctness gives $ar,ar'\in N$; Lemma~\ref{lem:ptld-collapse} therefore gives $r\theta r'$.  Hence $R_\alpha=R_\beta$, and unique factorization gives $\alpha=\beta$.  Thus there is at most one valid rule for each ordered pair $(N,A)$, so $|\Vcal|\le q^2$.
\end{proof}

\begin{corollary}[Prime interval DAG]
\label{cor:prime-dag}
Under \textup{(R)}, let $C$ be a live non-unit relative class and let
$w\in C$.  Form the directed acyclic graph with vertices
$0,\ldots,|w|$, placing an edge $i\to j$ exactly when the substring
$w[i:j]$ belongs to a relative prime class.  Then every minimum-edge
path from $0$ to $|w|$ reads the unique exact prime factorization of
$C$.
\end{corollary}

\begin{proof}
Any path determines a prime sequence whose setwise product is contained in $C$, because the chosen substrings concatenate to $w\in C$ and quotient multiplication returns the class $C$.  A minimum-edge path cannot contain a length-at-least-two contiguous subpath with prime quotient: the corresponding whole interval would itself be a prime edge and would shorten the path.  Theorem~\ref{thm:irreducible-path} therefore applies.
\end{proof}

\begin{theorem}[Prime-sequence refinement]
\label{thm:prime-refinement}
Assume \textup{(R)}.  Let \(C=\overline{A_1\cdots A_k}\) be the exact prime factorization of a live class.  If $\beta=B_1\cdots B_m$ is any prime sequence with
\(\overline\beta\subseteq C,\)
then the direct valid-rule grammar of Section~\ref{sec:valid} derives
\(A_1\cdots A_k\Rightarrow^* B_1\cdots B_m.\)
\end{theorem}

\begin{proof}
By Lemma~\ref{lem:prime-contraction}, contract $\beta$ to a prime-irreducible sequence with quotient class $C$ and setwise product contained in $C$.  Theorem~\ref{thm:irreducible-path} identifies the terminal sequence with $A_1\cdots A_k$.  Reversing the contractions gives correct prime expansions, and the Valid-Rule Reduction Lemma replaces each by valid productions.
\end{proof}

\begin{example}[PTLD is strictly weaker than factor cancellation]
\label{ex:ptld-not-cancel}
Let $M=\{1,a,0\}$ with $a^2=0$ and $0$ absorbing.  Put $\Sigma=\{x\}$, $h(x)=a$, and $L=\{x^2\}$.  The only live relative prime is $X=\{x\}$ with type $a$.  Equations $qr=a$ have at most one right solution for each $q$, so PTLD holds.  But
\(a\cdot a=0=a\cdot0, \qquad a\ne0,\)
so factor-left-cancellation fails.
\end{example}

The noncancellative language $L_0$ below shows that UF does not imply VTD or PTLD.  The $36$-element quotient witness in Section~\ref{sec:36} will show the stronger separation UF$\not\Rightarrow$FRP.

\section{Bounded-defect localization}
\label{sec:defect}

The next mechanism yields unique exact factorization without cancellation.

\begin{definition}[Bounded concatenation-defect statistic]
A function $\nu:\Sigma^*\to\Nat$ has bounded concatenation defect $\Delta$ if
\[
  \nu(x)+\nu(y)
  \le \nu(xy)
  \le \nu(x)+\nu(y)+\Delta
  \tag{BD}
\]
for all $x,y\in\Sigma^*$.
\end{definition}

Iterating (BD) gives
\begin{equation}
  \sum_{i=1}^k\nu(x_i)
  \le
  \nu(x_1\cdots x_k)
  \le
  \sum_{i=1}^k\nu(x_i)+(k-1)\Delta.
  \label{eq:defect-iteration}
\end{equation}

For a relative class $X$ define its defect spectrum
\(\Spec_\nu(X):=\{\nu(w):w\in X\}, \qquad \underline\nu(X):=\min\Spec_\nu(X).\)
Call a prime $P$ \emph{neutral} if $\Spec_\nu(P)=\{0\}$ and \emph{positive} if $\underline\nu(P)\ge1$.

\begin{theorem}[Single-carrier lemma]
\label{thm:single-carrier}
Assume every relative prime is neutral or positive.  Let $X$ be a relative class such that
\(\underline\nu(X)=1\)
and $\Spec_\nu(X)$ is unbounded.  Then every exact prime factorization of $X$ contains exactly one positive prime.  Moreover, that positive prime has unbounded defect spectrum.
\end{theorem}

\begin{proof}
If a factorization $X=P_1\cdots P_k$ contains no positive prime, then every factor has defect zero, so \eqref{eq:defect-iteration} gives a uniform bound $(k-1)\Delta$ on $\Spec_\nu(X)$, contradiction.  If it contains at least two positive primes, the lower bound in \eqref{eq:defect-iteration} gives $\nu(w)\ge2$ for every product word, contradicting $\underline\nu(X)=1$.  Hence there is exactly one positive prime.  If its spectrum were bounded, then \eqref{eq:defect-iteration} would again bound the spectrum of the entire product.
\end{proof}

For a fixed finite word $z$, the occurrence count $\nu_z(w)=\#_z(w)$ satisfies (BD) with $\Delta=|z|-1$, since only occurrences crossing the concatenation boundary can be newly created.  The example below uses $z=ba$ and $\Delta=1$.

\section{The nonregular example \texorpdfstring{$L_0=\{a^n b^n:n\ge0\}^*$}{The nonregular example L0}}
\label{sec:L0}

Put
\(K_n:=a^n b^n\quad(n\ge1), \qquad K:=\{K_n:n\ge1\}, \qquad L_0:=K^*.\)
For $w\in\{a,b\}^*$ let $\fst(w),\lst(w)\in\{\bot,a,b\}$ be the
first and last letters, and let $\operatorname{ba}(w)$ record whether
$w$ contains the factor $ba$.  We use the finite morphism
\(h(w):=(\fst(w),\lst(w),\operatorname{ba}(w)),\)
with the usual concatenation product on first/last symbols and the
$ba$-bit.  Moreover,
\[
  h(\eps)=(\bot,\bot,0),
\]
whereas every nonempty word has first and last symbols in $\{a,b\}$.  Hence $h^{-1}(h(\eps))=\{\eps\}$, so $[\eps]_{L_0,h}=\{\eps\}$ and $(L_0,h)$ is unit-separated.  We prove $h$-substitutability here, together with the relative-class, factorization, and valid-rule analysis, so that this boundary example is self-contained.

All claims about $L_0$ used later are proved analytically in this
section; no finite computation is used.

We shall repeatedly use the involution
\(\iota(w):=\sigma(w^{\mathrm R})\), where $w^{\mathrm R}$ is reversal
and $\sigma(a)=b$, $\sigma(b)=a$.  It preserves $L_0$, reverses
concatenation, and preserves the relative congruence.  Hence every
statement about a left fringe has a right-fringe dual.

For a word $x=x_1\cdots x_n$, write
\(\beta_x(i):=|x_1\cdots x_i|_a-|x_1\cdots x_i|_b.\)
\begin{lemma}[Block-height characterization]
\label{lem:L0-height-characterization}
A word $w\in\{a,b\}^*$ belongs to $L_0$ if and only if
\begin{enumerate}[label=\textup{(\roman*)},nosep]
\item $\beta_w(i)\ge0$ for every prefix and $\beta_w(|w|)=0$;
\item every occurrence $w_iw_{i+1}=ba$ satisfies $\beta_w(i)=0$.
\end{enumerate}
Consequently, if a live factor contains $ba$, all of its internal
$ba$ events occur at one common relative height.
\end{lemma}

\begin{proof}
Every block $K_n=a^nb^n$ starts at height zero, stays nonnegative, and
returns to zero; the only $ba$ factors of a product of such blocks are
block boundaries.  This proves necessity.  Conversely, split $w$ at
all zero-height positions.  Between two consecutive such positions the
relative height is positive at every proper prefix.  Condition
\textup{(ii)} forbids $ba$ inside such a segment, so the segment lies in
$a^*b^*$.  Its total balance is zero, hence it is some $K_n$.  Thus
$w\in K^*=L_0$.  For the final statement, embed the live factor $x$ in
$uxv\in L_0$.  Every internal $ba$ of $x$ must occur at absolute height
zero, hence at the same relative height $-\beta_u(|u|)$ inside $x$.
\end{proof}

\begin{proposition}[$h$-substitutability of $L_0$]
\label{prop:L0-h-substitutable}
The language $L_0$ is $h$-substitutable.
\end{proposition}

\begin{proof}
Write $\beta(x)=|x|_a-|x|_b$.  If a live factor $x$ contains $ba$, Lemma~\ref{lem:L0-height-characterization} shows that all its internal $ba$ events occur at one relative height; denote it by $c(x)$.  If $x$ contains no $ba$, put $c(x)=\bot$.

We first record a consequence of the height characterization.  For every live factor $x$, its two-sided distribution is determined by the triple
\(\bigl(h(x),\beta(x),c(x)\bigr).\)
Indeed, for arbitrary $\ell,r$, the criterion for $\ell x r\in L_0$ consists of total balance zero, nonnegativity of all prefixes, and height zero at every $ba$ boundary.  The total-balance contribution of $x$ is $\beta(x)$.  If $c(x)=\bot$, then $x\in a^*b^*$ and its least relative prefix height is determined by $\beta(x)$ together with its first and last symbols.  If $c(x)\ne\bot$, all internal $ba$ events occur at height $c(x)$ and the least relative prefix height is $c(x)$; the possible $ba$ events across the two external boundaries are determined by the first and last symbols recorded by $h(x)$.  Hence every clause of the characterization depends only on the displayed triple and on $\ell,r$.

Now suppose that $h(x)=h(y)$ and that $uxv,uyv\in L_0$.  Total balance gives $\beta(x)=\beta(y)$.  The $ba$-bit in $h$ says that either neither factor contains $ba$ or both do.  In the latter case, acceptance of the two displayed words forces every internal $ba$ event to absolute height zero, whence
\(c(x)=-\beta(u)=c(y).\)
Thus $x$ and $y$ have the same determining triple and therefore the same two-sided distribution.  This is exactly $h$-substitutability.
\end{proof}

\subsection{Complete normal forms for live relative classes}

For $r,s\ge1$ and $d\in\mathbb Z$ put
\(A_r:=\{a^r\}, \qquad B_s:=\{b^s\},\)
\(C_d:=\{a^i b^j:i,j\ge1,\ i-j=d\}.\)
For $ba$-containing factors define
\begin{equation}
\begin{aligned}
  X_{c,e}&:=b^cL_0a^e && (c,e\ge1),\\
  Y_{c,e}&:=b^cL_0a^eK && (c\ge1,e\ge0),\\
  Z_{c,e}&:=Kb^cL_0a^e && (c\ge0,e\ge1),\\
  W_{c,e}&:=Kb^cL_0a^eK && (c,e\ge0).
\end{aligned}
\label{eq:L0-families}
\end{equation}

\begin{lemma}[Live-class normal forms]
\label{lem:L0-normal-forms}
The live relative classes of $(L_0,h)$ are exactly
\[
  [\eps]_\theta,\qquad A_r\ (r\ge1),\qquad B_s\ (s\ge1),\qquad C_d\ (d\in\mathbb Z),
\]
together with the four families in \eqref{eq:L0-families}.  Every word in one of
$X_{c,e},Y_{c,e},Z_{c,e},W_{c,e}$ has total balance $e-c$, and every
internal occurrence of $ba$ occurs at relative height $-c$.
\end{lemma}

\begin{proof}
Let $x$ be a live factor and choose an occurrence of $x$ inside a word
\(K_{n_1}\cdots K_{n_t}\in L_0.\)
If $x$ contains no $ba$, then $x\in a^*b^*$.  Thus $x$ is either a
pure power $a^r$, a pure power $b^s$, or $a^i b^j$ with $i,j\ge1$.
The first two cases give $A_r$ and $B_s$.  In the third case the
relative class is indexed by the balance $d=i-j$, giving $C_d$.

Suppose now that $x$ contains $ba$.  At its left edge, either the
occurrence starts in the descending $b$-part of a block, contributing
$b^c$ with $c\ge1$, or it starts in the ascending $a$-part of a block.
In the latter case the remainder of that block has the form
\(a^r b^n=K_r b^{n-r},\)
so the left fringe is $Kb^c$ with $c\ge0$.  Dually, at the right edge
the occurrence either ends in an ascending $a$-part, giving $a^e$
with $e\ge1$, or in a descending $b$-part, in which case
\(a^n b^r=a^{n-r}K_r\)
and the right fringe is $a^eK$ with $e\ge0$.  All complete blocks
strictly between the two fringes form an arbitrary word of $L_0$.
This gives exactly the four forms in \eqref{eq:L0-families}.

It remains to verify that the displayed sets are precisely relative
classes.  For $C_d$, all words have the same $h$-type.  If $d\ge0$,
then
\(x b^d\in K\subseteq L_0 \qquad(x\in C_d),\)
while for $d<0$,
\(a^{-d}x\in K.\)
Hence all members of $C_d$ share an accepting context and are
syntactically congruent by $h$-substitutability.  The same context
distinguishes $C_d$ from $C_{d'}$ when $d'\ne d$.  Likewise, $A_r$
and $A_{r'}$ are distinguished by the right context $b^r$, and
$B_s,B_{s'}$ by the left context $a^s$.

For each of the four $ba$-families, every member shares the context
\((a^c,b^e).\)
Indeed, for $z\in L_0$ and $K_n\in K$,
\[
\begin{aligned}
  a^c(b^cza^e)b^e&=K_c zK_e,\\
  a^c(b^cza^eK_n)b^e&=K_c zK_{e+n},\\
  a^c(K_n b^cza^e)b^e&=K_{c+n}zK_e,
\end{aligned}
\]
and the fourth case combines the last two identities.  Thus every
set in \eqref{eq:L0-families} lies in one relative class.

The first/last/$ba$ observation separates the four families from one
another.  Within one fixed family, every $ba$ event occurs after a
relative balance drop of exactly $c$, hence at height $-c$, and every
word has total balance $e-c$.  If a word from parameters $(c',e')$
were accepted in the context $(a^c,b^e)$, its first internal $ba$
event would occur at absolute height $c-c'$ and therefore, by Lemma~\ref{lem:L0-height-characterization}, at height zero.  Hence $c=c'$.  The total
balance of the completed word then forces $e=e'$.  Thus different
parameters give different relative classes, completing the
classification.
\end{proof}

\subsection{Seven primes and exact fringe identities}

Set
\[
\begin{aligned}
A&:=A_1=\{a\}, & B&:=B_1=\{b\}, & C&:=C_0=K,\\
D&:=X_{1,1}=bL_0a, & E&:=Y_{1,0}=bL_0K=bK^+,\\
F&:=Z_{0,1}=KL_0a=K^+a, & G&:=W_{0,0}=KL_0K=K^{\ge2}.
\end{aligned}
\]
The involution satisfies
\(\iota(A)=B,\ \iota(B)=A,\ \iota(E)=F,\ \iota(F)=E\), while
$C,D,G$ are fixed.

\begin{lemma}[Exact fringe identities]
\label{lem:L0-fringe-identities}
The no-$ba$ classes satisfy
\(A_r=A^r, \qquad B_s=B^s,\)
and
\begin{equation}
  C_d=
  \begin{cases}
  A^dC,&d>0,\\
  C,&d=0,\\
  CB^{-d},&d<0.
  \end{cases}
  \label{eq:L0-Cd}
\end{equation}
The $ba$-classes satisfy the exact setwise identities
\begin{equation}
  X_{c,e}=B^{c-1}DA^{e-1},
  \label{eq:L0-X}
\end{equation}
\begin{equation}
  Y_{c,e}=
  \begin{cases}
  B^{c-1}E,&e=0,\\
  B^{c-1}DA^{e-1}C,&e\ge1,
  \end{cases}
  \label{eq:L0-Y}
\end{equation}
\begin{equation}
  Z_{c,e}=
  \begin{cases}
  FA^{e-1},&c=0,\\
  CB^{c-1}DA^{e-1},&c\ge1,
  \end{cases}
  \label{eq:L0-Z}
\end{equation}
and
\begin{equation}
  W_{c,e}=
  \begin{cases}
  G,&c=e=0,\\
  FA^{e-1}C,&c=0,e\ge1,\\
  CB^{c-1}E,&c\ge1,e=0,\\
  CB^{c-1}DA^{e-1}C,&c,e\ge1.
  \end{cases}
  \label{eq:L0-W}
\end{equation}
\end{lemma}

\begin{proof}
Each identity follows by multiplying the defining sets.  For example,
\(B^{c-1}DA^{e-1} =b^{c-1}(bL_0a)a^{e-1} =b^cL_0a^e=X_{c,e},\)
and
\(CB^{c-1}DA^{e-1}C =Kb^cL_0a^eK=W_{c,e}.\)
The remaining cases are identical one-line calculations; the identities
for $C_d$ use $A^dK=\{a^{n+d}b^n:n\ge1\}$ and its right-hand dual.
\end{proof}

\begin{theorem}[Prime classification for $L_0$]
\label{thm:L0-primes}
The live relative prime spectrum is exactly
\(\Pcal_{L_0,h}=\{A,B,C,D,E,F,G\}.\)
\end{theorem}

\begin{proof}
Lemma~\ref{lem:L0-fringe-identities} makes every listed live class
other than $A,B,C,D,E,F,G$ an exact product of at least two non-unit
classes.  It remains to prove that these seven classes are prime.
The singleton classes $A$ and $B$ are immediate.

For the remaining five classes we use Lemma~\ref{lem:shortest-word-cut}.
The shortest word of $C$ is $ab$.  Its only nontrivial cut gives
$A\cdot B=\{ab\}\ne C$ (for instance $aabb\in C$), so $C$ is prime.
The shortest word of $D$ is $ba$, and
$B\cdot A=\{ba\}\ne D$ (for instance $baba\in D$).

The shortest word of $E$ is $bab$.  Its two cuts give
\(B\cdot C=bK\subsetneq bK^+=E\)
and
\(D\cdot B=bL_0ab\subsetneq bK^+=E.\)
For example $babab$ witnesses strictness of the first inclusion and
$baabb$ witnesses strictness of the second.  Hence $E$ is prime.

Since $\iota(E)=F$ and $\iota$ preserves exact decomposability,
primality of $F$ follows from primality of $E$.  Finally, the shortest
word of $G$ is $abab$.  Its three cuts give
\(A\cdot E=abK^+\subsetneq K^{\ge2},\)
\(C\cdot C=K^2\subsetneq K^{\ge2}, \qquad F\cdot B=K^+ab\subsetneq K^{\ge2}.\)
These inclusions are strict: for example
\(K_2K_1=aabbab\in G\setminus A E, \qquad K_1K_1K_1=ababab\in G\setminus C C,\)
and
\(K_1K_2=abaabb\in G\setminus F B.\)
No nontrivial cut of a shortest word can therefore induce an exact
binary product, so $G$ is prime.
\end{proof}

\subsection{Unique factorization without cancellation}

Let
\(\nu(w):=\#_{ba}(w).\)
Then
\(\nu(xy)=\nu(x)+\nu(y)+[\lst(x)=b\ \&\ \fst(y)=a],\)
so $\Delta=1$ in \textup{(BD)}.  The primes $A,B,C$ are neutral, while
\(\Spec_\nu(D)=\Spec_\nu(E)=\Spec_\nu(F)=\Spec_\nu(G) =\{1,2,3,\ldots\}.\)
Every class in \eqref{eq:L0-families} has the same unbounded positive spectrum
with minimum one.  Hence Theorem~\ref{thm:single-carrier} implies that every
exact factorization of a $ba$-class contains exactly one of
$D,E,F,G$.

\begin{lemma}[Neutral-fringe lemma]
\label{lem:L0-neutral-fringe}
A sequence over $\{A,B,C\}$ whose factor boundaries introduce no new
$ba$ has the form
\(A^pC^\epsilon B^q, \qquad p,q\ge0, \quad \epsilon\in\{0,1\}.\)
\end{lemma}

\begin{proof}
A $B$ cannot be followed by $A$ or $C$, because both begin with $a$.
A $C$ cannot be followed by $A$ or by another $C$, because $C$ ends
with $b$.  Hence all $A$'s precede the optional $C$, and all $B$'s
follow it; at most one $C$ can occur.
\end{proof}

\begin{lemma}[Carrier-fringe rigidity]
\label{lem:L0-carrier-fringe}
Let $T$ be one of the classes in \eqref{eq:L0-families}, and suppose
$T=P_1\cdots P_k$ is an exact prime factorization.  Then its unique
positive carrier and its neutral fringes are exactly those displayed
in \eqref{eq:L0-X}--\eqref{eq:L0-W}.
\end{lemma}

\begin{proof}
Because $T$ contains words with exactly one $ba$, no boundary between
adjacent prime factors may itself be a $ba$ boundary: otherwise every
word of the exact product would contain at least the carrier event and
the additional boundary event.  Thus Lemma~\ref{lem:L0-neutral-fringe}
applies to the neutral factors on either side of the unique carrier.

The classes $D,E$ have all their internal $ba$ events at relative
height $-1$, whereas $F,G$ have them at height $0$.  A neutral prefix
changes that height only by its total balance.  The target class in
\eqref{eq:L0-families} has event height $-c$.  If the carrier is $F$ or $G$, the
no-$ba$ condition forces every preceding neutral factor to be an $A$;
the event-height equation then reads $p=-c$, so $p=c=0$.  Thus $F,G$
can occur only when $c=0$, with no neutral prefix.

If the carrier is $D$ or $E$, a no-$ba$ neutral prefix has the form
$A^pC^\epsilon B^q$.  When the target begins with $b$, necessarily
$p=\epsilon=0$.  When it begins with $a$, exact setwise equality must
realize an arbitrary leading block $K_n$ ($n\ge1$).  A fixed string of
$A$'s and $B$'s cannot supply this free block; the only possible neutral
source is one copy of $C=K$.  Moreover $p>0$ would force every initial
$a$-run to have length at least $p+1$, missing the target words whose
free leading block is $K_1$.  Hence $p=0$ and $\epsilon=1$.  In either
case the event-height equation is
\(-q-1=-c,\)
so $q=c-1$.  Therefore the left fringe is empty when $c=0$, is
$B^{c-1}$ for an initial-$b$ class, and is $CB^{c-1}$ when a free
leading $K$ is present.

The right-fringe statement is the $\iota$-dual of the preceding
left-fringe argument.  Thus, if the carrier is $E$ or $G$, exactness
forces $e=0$ and no extra neutral suffix; if the carrier is $D$ or $F$,
the suffix is $A^{e-1}$ when the target ends in $a$, and $A^{e-1}C$
when the target has a free trailing $K$.

Consequently the carrier is determined by whether $c$ and $e$ vanish:
\((c>0,e>0):D, \quad (c>0,e=0):E, \quad (c=0,e>0):F, \quad (c=e=0):G,\)
and the surrounding neutral factors are exactly those in
\eqref{eq:L0-X}--\eqref{eq:L0-W}.
\end{proof}

\begin{theorem}[Noncancellative unique relative factorization]
\label{thm:L0-UF}
Every live non-unit relative class of $L_0$ has a unique exact prime
factorization, namely the decompositions in \eqref{eq:L0-Cd}--\eqref{eq:L0-W}.
\end{theorem}

\begin{proof}
For a singleton class $A_r$ or $B_s$, every exact prime factor must
have the same one-letter orientation, so shortest-length additivity
forces $A_r=A^r$ and $B_s=B^s$.

Now let $C_d$ be a mixed no-$ba$ class.  Any exact factorization uses
only neutral primes, and no factor boundary may create $ba$.  By
Lemma~\ref{lem:L0-neutral-fringe} its prime sequence has the form
$A^pC^\epsilon B^q$.  Since $C_d$ is infinite, $\epsilon=1$; otherwise
the product of the singleton classes $A$ and $B$ would itself be a
singleton.  Balance gives $p-q=d$.  Moreover
\(\ell(C_d)=|d|+2\)
and Lemma~\ref{lem:shortest-word-cut} gives $p+q+2=|d|+2$.  Hence
\(p=\max(d,0),\qquad q=\max(-d,0),\)
which is exactly \eqref{eq:L0-Cd}.

For a $ba$-class, Theorem~\ref{thm:single-carrier} gives exactly one positive
prime and Lemma~\ref{lem:L0-carrier-fringe} forces the carrier and both
neutral fringes uniquely.  The identities are exact by
Lemma~\ref{lem:L0-fringe-identities}.
\end{proof}

The observer monoid is nevertheless not factor-cancellative: the
first/last/$ba$ type of $a$ equals that of $aa$, so
\(h(a)h(\eps)=h(a)=h(a)h(a), \qquad h(\eps)\ne h(a).\)
Thus cancellation is sufficient for the Clark-style rigidity mechanism
but is not necessary for unique exact relative factorization.

\subsection{The thirty-four valid productions}

For a prime $P$ define interface bits
\(c(P):=[\fst(P)=b], \qquad e(P):=[\lst(P)=a].\)
Their values, together with the balance
$d(P)=e(P)-c(P)$, are:

\begin{center}
\begin{tabular}{c c c c c}
\toprule
Prime & shortest representative & $c(P)$ & $e(P)$ & $d(P)$\\
\midrule
$A$ & $a$ & $0$ & $1$ & $+1$\\
$B$ & $b$ & $1$ & $0$ & $-1$\\
$C$ & $ab$ & $0$ & $0$ & $0$\\
$D$ & $ba$ & $1$ & $1$ & $0$\\
$E$ & $bab$ & $1$ & $0$ & $-1$\\
$F$ & $aba$ & $0$ & $1$ & $+1$\\
$G$ & $abab$ & $0$ & $0$ & $0$\\
\bottomrule
\end{tabular}
\end{center}

\begin{lemma}[Binary interface lemma]
\label{lem:L0-binary-interface}
For primes $P,Q\in\Pcal_{L_0,h}$, $[PQ]_\theta$ is prime if and only if $e(P)=c(Q)$.
\end{lemma}

\begin{proof}
Take shortest representatives of $P$ and $Q$.  Lemma~\ref{lem:L0-normal-forms} determines the quotient class and Theorem~\ref{thm:L0-primes} its primality.  The prime cases are exactly the nine pairs with $e(P)=c(Q)=1$, namely
\(P\in\{A,D,F\}\) and \(Q\in\{B,D,E\}\), and the sixteen pairs with $e(P)=c(Q)=0$.
\end{proof}

Every correct length-two production is automatically valid, so
Lemma~\ref{lem:L0-binary-interface} gives exactly twenty-five binary
valid productions.  To control longer rules we use the common event
height of the $ba$-classes.

\begin{lemma}[Internal-$C$ lemma]
\label{lem:L0-internal-C}
Let
\(\alpha=P_1\cdots P_n, \qquad n\ge3,\)
be a prime sequence whose quotient product is a prime.  If every
adjacent pair has nonprime quotient product, then some internal symbol
$P_i$ ($1<i<n$) is $C$.
\end{lemma}

\begin{proof}
Write $c_i=c(P_i)$, $e_i=e(P_i)$, and $d_i=e_i-c_i$.  By the binary
interface lemma, adjacent nonprimality gives
\(c_{i+1}=1-e_i \qquad(1\le i<n).\)
Let
\(s_{i-1}:=\sum_{j<i}d_j, \qquad \eta_i:=s_{i-1}-c_i.\)
The quantity $\eta_i$ is the whole-word relative height of the internal $ba$ event carried by $P_i$ when such an event is present.  By the prime table above, the balance of every prime $R$ is $e(R)-c(R)$.  Then
\(\eta_{i+1}-\eta_i =d_i-c_{i+1}+c_i =2e_i-1\in\{-1,+1\}.\)
The first symbol of the quotient prime is the first symbol of $P_1$
and its last symbol is the last symbol of $P_n$.  Correctness therefore gives
\(\sum_{i=1}^n d_i=e_n-c_1.\)
Equivalently,
\(\eta_n=\eta_1.\)
Thus $\eta_1,\ldots,\eta_n$ is a nontrivial closed walk with unit steps.

If the quotient prime is one of $A,B,C$, then no $P_i$ can be a carrier,
because a carrier already has $\operatorname{ba}=1$ and the mismatching
interfaces create no new $ba$ boundary.  If the quotient prime is one
of $D,E,F,G$, every carrier $P_i$ in the sequence has all of its
internal $ba$ events at whole-word relative height $\eta_i$.  The target
prime has all such events at height $\eta_1=-c_1$, so every carrier must
occur at a position with $\eta_i=\eta_1$.

The closed walk cannot go below its initial level.  Otherwise an
internal strict local minimum would have incoming step $-1$ and outgoing
step $+1$, hence $c_i=e_i=1$.  The only such prime is $D$, a carrier at
a level different from $\eta_1$, contradiction (and in the neutral
target case carriers are impossible altogether).  Therefore the walk
has a positive excursion.  At an internal strict local maximum, the
incoming step is $+1$ and the outgoing step is $-1$, so
$c_i=e_i=0$.  Hence $P_i\in\{C,G\}$.  The level is strictly above
$\eta_1$, so $G$ cannot occur there as a carrier; in the neutral case it
is excluded for the same reason as above.  Thus $P_i=C$.
\end{proof}

\begin{theorem}[FRP for $L_0$]
\label{thm:L0-FRP}
For the first/last/$ba$ observation,
\(|\Pcal_{L_0,h}|=7, \qquad |\Vcal_{L_0,h}|=34, \qquad \rho(L_0,h)=3.\)
\end{theorem}

\begin{proof}
A valid right-hand side of length at least three cannot contain a
prime-valued adjacent pair, since such a pair would be a proper
pleonastic block.  Hence all adjacent interfaces mismatch, and
Lemma~\ref{lem:L0-internal-C} supplies an internal $C$.
If $P_{i-1}CP_{i+1}$ surrounds such an occurrence, mismatch on the left
forces $e(P_{i-1})=1$, so
\(P_{i-1}\in\{A,D,F\},\)
while mismatch on the right forces $c(P_{i+1})=1$, so
\(P_{i+1}\in\{B,D,E\}.\)
The nine possible triples in $\{A,D,F\}C\{B,D,E\}$ all return to primes by the normal forms; they are exactly the ternary rules displayed in the complete grammar below.  Hence a right-hand side of length greater than three contains a proper prime-returning triple and is pleonastic, while length three yields exactly those nine valid rules.  Together with the twenty-five binary rules this gives $|\Vcal|=25+9=34$ and $\rho=3$.
\end{proof}

The thirty-four branching rules are therefore
\[
\begin{aligned}
C&\to AB\mid ACB,\\
D&\to BA\mid BF\mid DD\mid EA\mid EF\mid DCD,\\
E&\to BC\mid BG\mid DB\mid DE\mid EC\mid EG\mid DCB\mid DCE,\\
F&\to AD\mid CA\mid CF\mid FD\mid GA\mid GF\mid ACD\mid FCD,\\
G&\to AE\mid CC\mid CG\mid FB\mid FE\mid GC\mid GG\mid ACE\mid FCB\mid FCE.
\end{aligned}
\]
Together with $A\to a$, $B\to b$, and
\(S\to\eps\mid C\mid G,\)
these form the direct relative-prime grammar for $L_0$.

\begin{proposition}[Tail saturation for $L_0$]
Every valid production in the direct relative-prime grammar for $L_0$
is tail-saturated.
\end{proposition}

\begin{proof}
Every binary valid production is automatically tail-saturated.  In the
nine ternary rules, the suffix after the first prime is one of
\(CB,\qquad CD,\qquad CE.\)
The normal forms give the exact identities
\(C\cdot B=C_{-1}, \qquad C\cdot D=Z_{1,1}, \qquad C\cdot E=W_{1,0}.\)
Thus every ternary tail is also exact.
\end{proof}

\begin{remark}[UF does not imply VTD]
The two valid rules
\(D\to BA, \qquad D\to BF\)
have the same left-hand prime and the same first right-hand prime but
different tails.  Hence $L_0$ has exact UF and tail saturation while
failing VTD and PTLD.
\end{remark}

\section{Finite-state relative presentation}
\label{sec:fsrp}

\subsection{Definition and finite-quotient regularity}

FRP demands a finite \emph{list} of valid rules.  For finite grammar construction this is stronger than necessary: an infinite valid family can be represented by a finite-state controller.

For each prime $P$, define
\(\Corr_P:=\{\alpha\in\Pcal^{\ge2}:[\overline\alpha]=P\}, \qquad \Corr:=\bigcup_P\Corr_P.\)
Then the pleonastic words are exactly
\begin{equation}
  \Pleon
  =
  \Pcal^+\Corr\Pcal^*
  \ \cup\
  \Pcal^*\Corr\Pcal^+,
  \label{eq:pleonastic-language}
\end{equation}
and
\[
  \Val_P=\Corr_P\setminus\Pleon.
\]

\begin{definition}[Finite-state relative presentation property]
A pair $(L,\theta)$ has the \emph{finite-state relative presentation property} (FSRP) if $|\Pcal|<\infty$ and every $\Val_P\subseteq\Pcal^{\ge2}$ is regular.  We write $\mathrm{FSRP}(L,\theta)$ for this property and abbreviate it by $\mathrm{FSRP}(L,h)$ when $\theta=\theta_{L,h}$.
\end{definition}

Clearly
\(\mathrm{FRP}\Longrightarrow\mathrm{FSRP}.\)
The converse fails by the $36$-element finite-quotient example in Section~\ref{sec:36}.

\begin{theorem}[Finite quotient implies FSRP]
\label{thm:finite-quotient-fsrp}
If $S=\Sigma^*/\theta$ is finite, then FSRP holds.
\end{theorem}

\begin{proof}
The evaluation morphism
\(\mu:\Pcal^*\to S\)
has finite codomain.  Hence
\(\Corr_P=\mu^{-1}(P)\cap\Pcal^{\ge2}\)
is regular for every prime $P$, so $\Corr$ is regular.  Equation~\eqref{eq:pleonastic-language} and closure of regular languages under concatenation, union, and difference imply regularity of every $\Val_P$.
\end{proof}

\begin{proposition}[Context-freeness before factor-minimalization]
\label{prop:corr-cfl}
Let $L$ be context-free and $h$-substitutable, and assume that the relative prime spectrum $\Pcal$ is finite.  Then every live relative class is context-free, every correct-return language $\Corr_P$ is context-free, and therefore $\Corr$ and $\Pleon$ are context-free.
\end{proposition}

\begin{proof}
Let $C=[u]_{L,h}$ be a live relative class and choose an accepting context $(x,y)$ with $xuy\in L$.  By $h$-substitutability, a word $v$ lies in $C$ exactly when it has the same observer type as $u$ and occurs in this same accepting context.  Hence
\[
  C=h^{-1}(h(u))\cap x^{-1}Ly^{-1},
\]
where $x^{-1}Ly^{-1}=\{v:xvy\in L\}$.  The second language is context-free because context-free languages are closed under left and right quotient by fixed words, while $h^{-1}(h(u))$ is regular because $h$ has finite codomain.  Thus $C$ is context-free.

For each prime $Q\in\Pcal$, choose a representative word $\omega(Q)\in Q$ and extend $Q\mapsto\omega(Q)$ to a morphism
\(\phi:\Pcal^*\to\Sigma^*.\)
For every prime sequence $\alpha$, the word $\phi(\alpha)$ lies in the quotient product class $[\overline\alpha]$.  Consequently
\[
  \Corr_P=\phi^{-1}(P)\cap\Pcal^{\ge2}.
\]
Since $P$ is context-free, inverse homomorphism and intersection with the regular language $\Pcal^{\ge2}$ show that $\Corr_P$ is context-free.  Finiteness of $\Pcal$ gives context-freeness of $\Corr=\bigcup_P\Corr_P$, and Equation~\eqref{eq:pleonastic-language} then gives context-freeness of $\Pleon$ by closure under concatenation with regular languages and finite union.
\end{proof}

\begin{remark}[The unresolved FSRP boundary]
Proposition~\ref{prop:corr-cfl} localizes the remaining issue.  Under the context-free, $h$-substitutable, finite-prime hypotheses, nonregularity cannot arise already at the level of correct returns: each $\Corr_P$ and the global pleonastic language are context-free.  FSRP asks whether the factor-minimal difference
\[
  \Val_P=\Corr_P\setminus\Pleon
\]
is regular.  Context-free languages are not closed under difference, so the proposition alone gives no regularity conclusion.  At present, we are not aware of an example in this setting for which some $\Val_P$ is nonregular; all explicit finite-prime examples in this paper satisfy FSRP.  Thus the possibility that FSRP follows automatically from context-freeness, $h$-substitutability, and finite relative prime spectrum is not ruled out here.
\end{remark}

The preceding finite-quotient regularity theorem is abstract: it follows immediately from quotient evaluation.  The following automaton is retained for the stronger algorithmic tasks of deciding FRP, computing the validity radius, and extracting pumping certificates.

\subsection{Prime-avoidance automata for finite quotients}
\label{sec:avoid}

When the relative quotient
\(S:=\Sigma^*/\theta\)
is finite, valid RHSs can be recognized by a finite automaton intrinsic to quotient multiplication.

For a nonempty prime word
\(x=P_1\cdots P_j\)
define its total product
\(t(x):=P_1*\cdots*P_j\)
and the set of products of proper nonempty suffixes beginning after the first symbol,
\(U(x):=\{P_i*\cdots*P_j:2\le i\le j\}.\)
For $j=1$, set $U(x)=\varnothing$.  A state is a pair $(t,U)$.

Appending a prime $Q$ changes the suffix profile to
\(U'=\{Q\}\cup\{u*Q:u\in U\}.\)
A \emph{safe transition}
\((t,U)\xrightarrow{Q}(t*Q,U')\)
is permitted if
\begin{equation}
  t*Q\notin\Pcal,
  \qquad
  \{u*Q:u\in U\}\cap\Pcal=\varnothing.
  \label{eq:safe-transition}
\end{equation}
A \emph{final edge} labelled $Q$ is permitted if
\begin{equation}
  t*Q\in\Pcal,
  \qquad
  \{u*Q:u\in U\}\cap\Pcal=\varnothing.
  \label{eq:final-transition}
\end{equation}
The initial states are $(P,\varnothing)$ for $P\in\Pcal$.

\begin{theorem}[Prime-avoidance automaton]
\label{thm:prime-avoidance}
A prime word $P_1\cdots P_k$, $k\ge2$, is a valid RHS iff starting from $(P_1,\varnothing)$, the symbols $P_2,\ldots,P_{k-1}$ can be read along safe transitions and $P_k$ along a final edge.
\end{theorem}

\begin{proof}
Inductively, the state records the quotient products of every suffix that could become a proper contiguous block ending at the current position.  Condition~\eqref{eq:safe-transition} prevents both the whole current prefix and every newly completed proper suffix of length at least two from being prime-valued.  Condition~\eqref{eq:final-transition} changes only the whole word requirement: at the last symbol the total product must return to a prime, while every proper suffix must remain nonprime.  Earlier proper intervals were already ruled out when their right endpoints were processed.
\end{proof}

The automaton recognizes the global valid language $\Val$.  To isolate $\Val_P$, retain only final edges whose total quotient product is the target prime $P$.

The state set has size at most
\(|S|2^{|S|},\)
so it is finite whenever $S$ is finite.

\begin{theorem}[Productive-cycle criterion]
Assume $S$ is finite.  Restrict the safe-transition graph to states reachable from an initial state and from which a final edge is reachable.  Then
\(\mathrm{FRP} \iff \text{this productive safe graph is acyclic}.\)
\end{theorem}

\begin{proof}
A productive cycle can be iterated arbitrarily many times before taking a path to a final edge, giving arbitrarily long valid RHSs.  Conversely, if the productive graph is acyclic, every productive safe path has bounded length, and the finite prime alphabet then yields only finitely many valid words.
\end{proof}

\begin{corollary}
Suppose the finite relative quotient is given effectively by its multiplication table together with the finite set of live prime classes.  Then FRP is decidable from this finite data, and the validity radius is computable.
\end{corollary}

\begin{proof}
Construct the finite prime-avoidance automaton from the multiplication table and the distinguished prime subset, and test the productive subgraph for a directed cycle.  In the acyclic case its longest productive path gives the validity radius.
\end{proof}

\begin{corollary}[Pumping certificate]
\label{cor:pumping-certificate}
If the finite-quotient pair fails FRP, there exist prime words $u,v,w$ with $v\neq\eps$ such that $uv^n w$ is valid for every $n\ge0$.
\end{corollary}

\begin{proof}
Take a path into a productive cycle, one nonempty traversal of the cycle, and a path from it to a final edge; their labels give $u,v,w$.
\end{proof}

Section~\ref{sec:36} gives an explicit one-state productive loop for the family $P_{ab}\to P_aP_q^mP_b$.

\subsection{Canonical residual controller}

For each prime $P$ and prefix $u\in\Pcal^*$, take the left quotient
\(u^{-1}\Val_P:=\{v:uv\in\Val_P\}.\)
Let
\[
  \Qcal_{L,\theta}
  :=
  \{u^{-1}\Val_P:
  P\in\Pcal,\ u\in\Pcal^*,\ u^{-1}\Val_P\neq\varnothing\}.
\]

\begin{definition}[Presentation state complexity]
Define
\(\kappa(L,\theta):=|\Qcal_{L,\theta}|\in\Nat\cup\{\infty\}.\)
In the fixed-$h$ case write $\Qcal_{L,h}:=\Qcal_{L,\theta_{L,h}}$ and $\kappa(L,h):=\kappa(L,\theta_{L,h})$.
\end{definition}

\begin{theorem}[Residual characterization of FSRP]
Assuming $|\Pcal|<\infty$,
\[
  \mathrm{FSRP}(L,\theta)\iff\kappa(L,\theta)<\infty.
\]
Moreover, the nonempty residual languages above form the canonical minimal deterministic partial multi-start controller for the family $(\Val_P)_{P\in\Pcal}$.  Adjoining the empty residual $\varnothing$ as a single dead state gives the corresponding canonical complete controller.
\end{theorem}

\begin{proof}
This is the Myhill--Nerode theorem applied simultaneously to the finite family of regular languages $\Val_P$, taking residual languages themselves as states.  In the usual complete realization, every empty left quotient is represented by the single dead state $\varnothing$; deleting that state and the transitions entering it yields exactly the partial controller displayed above.  Minimality and canonicity therefore follow from equality of residual languages, not merely from state renaming.
\end{proof}

\begin{definition}[Canonical finite-state relative-prime presentation]
\label{def:canonical-fsrp}
For an FSRP pair $(L,h)$ define
\[
  \mathfrak G^{\mathrm{fs}}_{L,h}
  :=
  \bigl(
    \Pcal,\mathrm{Lex},\mathrm{Start},
    \Qcal_{L,h},\delta,\mathrm{Fin}
  \bigr),
\]
where $\mathrm{Lex}(P)=P\cap\Sigma$, $\mathrm{Start}$ records the exact prime decompositions of the accepting relative classes (and the empty start alternative when $\eps\in L$), $\Qcal_{L,h}$ is the shared set of nonempty residual languages defined above, $\delta(R,A)=A^{-1}R$ whenever this quotient is nonempty, and $\mathrm{Fin}$ records the residual states containing $\eps$.  This tuple is canonical because its controller states are residual languages themselves rather than arbitrarily named automaton states.
\end{definition}

\subsection{Compilation to an ordinary finite CFG}

Assume unit separation, finitely many accepting classes, and FSRP.  For each residual state $R\in\Qcal_{L,h}$ introduce an auxiliary nonterminal $C_R$.  If
\(A^{-1}R=R'\neq\varnothing,\)
include
\(C_R\to A C_{R'}.\)
If $\eps\in A^{-1}R$, also include the terminating controller step
\(C_R\to A.\)
For each prime $P$ with nonempty $\Val_P$, include
\(P\to C_{\Val_P}.\)
Add the lexical and start rules as in the direct construction.

\begin{theorem}[Finite-state canonical grammar theorem]
Under unit separation, finitely many accepting classes, finite prime spectrum, and FSRP, the residual-controller construction yields a finite CFG $G^{\mathrm{res}}_{L,h}$ satisfying
\(L(G^{\mathrm{res}}_{L,h},P)=P\)
for every live prime $P$, and therefore
\(L(G^{\mathrm{res}}_{L,h})=L.\)
\end{theorem}

\begin{proof}
The controller realizes exactly the valid right-hand sides.  Hence soundness and completeness reduce to the same arguments as for the direct grammar, using valid-rule reduction.  Finiteness follows from $|\Pcal|<\infty$ and $\kappa(L,h)<\infty$.
\end{proof}

\begin{theorem}[Prime-skeleton preservation]
\label{thm:prime-skeleton-preservation}
Contract every maximal chain of residual-controller nonterminals in a parse tree of $G^{\mathrm{res}}_{L,h}$.  The resulting tree is a parse tree of the direct grammar containing all valid productions.  Conversely, every direct valid-rule parse has a unique expansion through the deterministic residual controller.  Hence contraction induces a bijection between residual-controller parse trees and direct prime-skeleton parse trees.
\end{theorem}

\begin{proof}
A controller chain beginning with $P\to C_{\Val_P}$ and emitting a prime word $\alpha=A_1\cdots A_k$ is forced step by step by the deterministic residual transition $R\mapsto A^{-1}R$.  It terminates exactly when the residual reached after reading $\alpha$ contains $\eps$, equivalently when $\alpha\in\Val_P$.  Contracting the chain therefore yields precisely the direct valid production $P\to\alpha$.  Conversely, for every $\alpha\in\Val_P$, the same deterministic residual transitions give a unique controller chain from $C_{\Val_P}$ to a terminating state.  Applying this independently at every branching node gives mutually inverse expansion and contraction maps on parse trees.
\end{proof}

Thus FSRP is not merely weak language compression: it is a finite-state subdivision of the same intrinsic prime branching structure.  This is the sense in which FSRP remains suitable for a strong/canonical structural target, paralleling Clark's motivation for canonical grammars~\cite{Clark2013Trees}.

\section{Tail saturation and presentation-defect complexity}
\label{sec:saturation}

The semantic residual theorem shows that only finitely many relative tail classes can occur after a fixed pair of primes.  The remaining source of direct-rule infinitude is failure of a valid tail to saturate its own relative class.

\begin{definition}[Saturated and under-saturated valid rules]
Write a valid rule as
\(P\to A\alpha, \qquad P,A\in\Pcal,\ \alpha\in\Pcal^+.\)
It is \emph{tail-saturated} if
\(\overline\alpha=[\overline\alpha]_\theta,\)
and \emph{tail-under-saturated} otherwise.  Let $\Sat_P$ and $\Unsat_P$ be the corresponding right-hand-side languages for left-hand prime $P$.
\end{definition}

Thus VTE$_1$ is equivalent to $\Unsat_P=\varnothing$ for every prime $P$.  Every binary valid rule is automatically saturated, because its tail consists of one prime class.

For finite prime spectrum, put
\[
  q:=|\Pcal|,
  \qquad
  r_{\rm res}(L,h):=\max\Bigl(\{|\operatorname{Res}_\theta(A,P)|:A,P\in\Pcal\}\cup\{0\}\Bigr),
\]
and let
\[
  \mathcal R_{\rm rel}:=\bigcup_{A,P\in\Pcal}\operatorname{Res}_\theta(A,P).
\]
For $R\in\mathcal R_{\rm rel}$ define
\[
  \Ex(R):=\{\alpha\in\Pcal^+:\overline\alpha=R\},
  \qquad
  m_{\rm ex}(L,h):=\max\bigl(\{|\Ex(R)|:R\in\mathcal R_{\rm rel}\}\cup\{0\}\bigr).
\]
The quantity $m_{\rm ex}(L,h)$ is finite even without UF by Corollary~\ref{cor:finite-candidates}.

\begin{theorem}[Saturated-basis bound]
\label{thm:saturated-basis-bound}
Assume $L$ is $h$-substitutable and $|\Pcal|=q<\infty$.  Then the set of saturated valid rules is finite and
\(|\Vcal_{\mathrm{sat}}| \le q^2 r_{\rm res}(L,h)m_{\rm ex}(L,h).\)
\end{theorem}

\begin{proof}
For fixed $P,A$, a saturated tail is an exact prime decomposition of some $R\in\operatorname{Res}_\theta(A,P)$.  There are at most $r_{\rm res}(L,h)$ such residual classes and at most $m_{\rm ex}(L,h)$ exact decompositions of each.  Sum over the $q^2$ choices of $(P,A)$.
\end{proof}

Define the under-saturation radius and defect residual complexity by
\(\rho_\partial(L,h):=\sup\{|\alpha|:\alpha\in\Unsat_P\text{ for some }P\}\)
and
\[
  \kappa_\partial(L,h)
  :=
  \left|
  \left\{
    u^{-1}\Unsat_P:
    P\in\Pcal,\;
    u\in\Pcal^*,\;
    u^{-1}\Unsat_P\ne\varnothing
  \right\}
  \right|.
\]

\begin{corollary}[Saturation--defect decomposition]
Under the same hypotheses,
\[
\begin{aligned}
\mathrm{VTE}_1&\iff \Unsat_P=\varnothing\text{ for every }P,\\
\mathrm{FRP}&\iff \forall P,\ |\Unsat_P|<\infty\iff\rho_\partial(L,h)<\infty,\\
\mathrm{FSRP}&\iff \forall P,\ \Unsat_P\text{ is regular}\iff\kappa_\partial(L,h)<\infty.
\end{aligned}
\]
Moreover $\rho(L,h)<\infty\iff\rho_\partial(L,h)<\infty$.
\end{corollary}

\begin{proof}
For every $P$, $\Val_P=\Sat_P\,\dot\cup\,\Unsat_P$ and the saturated part is finite by Theorem~\ref{thm:saturated-basis-bound}.  Over a finite prime alphabet, finiteness is equivalent to bounded word length, and finite unions and differences preserve regularity.  The first equivalence is the definition of VTE$_1$.
\end{proof}

\begin{theorem}[Single-residual localization of FRP failure]
If $|\Pcal|<\infty$, $L$ is $h$-substitutable, and FRP fails, then there are fixed primes $P,A$ and a fixed live relative class $R\subseteq A\backslash P$ together with infinitely many pairwise distinct valid rules $P\to A\alpha_n$ such that $[\overline{\alpha_n}]_\theta=R$ and $\overline{\alpha_n}\subsetneq R$.  The lengths $|\alpha_n|$ are unbounded along a subsequence.
\end{theorem}

\begin{proof}
Only finitely many triples $(P,A,R)$ are possible by finite prime spectrum and finite residual splitting.  Since the saturated part is finite, infinitely many valid rules must be under-saturated.  Pigeonhole localization gives one fixed triple; bounded length over a finite prime alphabet would give only finitely many words.
\end{proof}

This motivates the \emph{relative presentation complexity profile} $\Pi(L,h):=(q,r_{\rm res},m_{\rm ex};\rho_\partial,\kappa_\partial)$.  The first three coordinates measure static relative geometry; the last two measure the dynamic structural defect.  Under PTLD and finite prime spectrum, Theorem~\ref{thm:finite-residual-splitting} and unique exact factorization give
\[
  r_{\rm res}\le1,\qquad m_{\rm ex}\le1,\qquad \rho_\partial=\kappa_\partial=0.
\]
Thus the static residual and exact-factor multiplicities collapse to at most one, while the dynamic under-saturation defect disappears entirely.

If FSRP holds but FRP fails, some $\Unsat_P$ is an infinite regular language, so the ordinary pumping lemma gives an infinite family $uv^nw$ of under-saturated valid words.  By the preceding theorem, an infinite subsequence can be localized to one residual class $R$.  When the relative quotient is finite, Corollary~\ref{cor:pumping-certificate} refines this to an explicit automaton-lasso certificate.

\section{Unique factorization without finite direct presentation: a 36-element quotient witness}
\label{sec:36}

We now show directly that exact-factor ambiguity is not responsible for
failure of FRP.  We deliberately choose the trivial target language
$L=\Sigma^+$, so that the separation cannot be attributed to complicated
language syntax: all nontrivial behavior is induced by the finite observer.
Let $Q_0=\{0,1,2,3\}$ and write transformations as
tuples of their values on $Q_0$.  Put
\(\tau_a=(1,2,0,0)\) and \(\tau_b=(1,3,1,0)\), and let
\(M_0:=\langle\tau_a,\tau_b\rangle\le T_4.\)
Finite enumeration gives $|M_0|=64$.  Define
\(h_0:\{a,b\}^*\to M_0\) by
$h_0(a)=\tau_a$ and $h_0(b)=\tau_b$.  Every nonempty product has rank
at most three whereas the identity has rank four, so
\(h_0^{-1}(1)=\{\eps\}.\)

Let $\sim_0$ be the least monoid congruence on $M_0$ containing
\(h_0(ab)\sim_0h_0(abaaab)\), and put
\(M_{36}:=M_0/{\sim_0}\) and \(h:=\pi_0\circ h_0\).
Again take $L:=\{a,b\}^+$.  Finite congruence closure gives
$|M_{36}|=36$, and the identity class has no nonempty preimage, so
$h^{-1}(1)=\{\eps\}$.  The following witness should therefore be read first as a finite monoid phenomenon and only secondarily as a language example.  Choosing $L=\Sigma^+$ removes all nontrivial syntactic distinctions between nonempty words and isolates the obstruction entirely in the finite observer quotient.  Indeed, since $L=\Sigma^+$, all nonempty words are syntactically congruent and $\theta_{L,h}=\ker h$ on nonempty words.
Thus the $35$ nonidentity quotient classes are exactly the live
non-unit relative classes.

\paragraph{Computer-assisted verification.}
The monoid generators, quotient congruence, and infinite-family identities are explicit.  The remaining finite counts and quotient-minimality assertions are checked exhaustively by the reproducible verifier described in Appendix~\ref{app:36-computation}.

\begin{theorem}[Finite structure of the $36$-element witness]
\label{thm:36-finite-structure}
For the pair $(L,h)$ above, unit separation and $h$-substitutability hold, the relative prime spectrum has fifteen elements, and every one of the thirty-five live non-unit relative classes has a unique exact relative prime factorization.
\end{theorem}

\begin{proof}
Unit separation follows from $h^{-1}(1)=\{\eps\}$.  Since $L=\Sigma^+$, all nonempty words are syntactically congruent.  If two equal-$h$-type words share an accepting context, then either both are nonempty, in which case they are syntactically congruent, or both are empty because $h^{-1}(1)=\{\eps\}$.  Thus $L$ is $h$-substitutable.

The remaining finite assertions are verified exhaustively.  Exact binary fiber-product equality yields the fifteen relative primes
\[
\begin{aligned}
\Pcal=\{&a,b,ab,ba,aba,baa,bab,aaab,abaa,abab,baaa,\\
&baba,aaaba,babaa,aaabaa\},
\end{aligned}
\]
where each word denotes its relative class.  Every quotient class has a shortest representative of length at most eight.  By Corollary~\ref{cor:finite-candidates}, cuts of these shortest representatives contain every exact prime factorization candidate.  Exhaustive regular-language equality testing leaves exactly one candidate for each of the thirty-five live non-unit classes; Appendix~\ref{app:36-computation} records the complete certificate.
\end{proof}

For the infinite valid family, put
\(P_a=[a]_\theta,\quad P_q=[baaa]_\theta,\quad P_b=[b]_\theta,\quad P_{ab}=[ab]_\theta.\)
All four are among the fifteen relative primes.

\begin{proposition}[Infinite valid family]
For every $m\ge0$,
\[
  \boxed{
  P_{ab}\to P_aP_q^mP_b
  }
\]
is valid.
\end{proposition}

\begin{proof}[Finite quotient certificate]
Let
\(\alpha=h(a),\qquad q=h(baaa),\qquad \beta=h(b),\qquad t=h(ab).\)
The finite multiplication table satisfies
\[
  q^2=q^3,
  \qquad
  \alpha q=\alpha q^2,
  \qquad
  q\beta=q^2\beta,
\]
and
\[
  \alpha\beta=t,
  \qquad
  \alpha q\beta=t,
  \qquad
  \alpha q^2\beta=t.
\]
Hence $\alpha q^m\beta=t$ for every $m\ge0$, proving correctness of the displayed family.  The same three stabilization identities make the proper intervals uniform in $m$: every proper block of length at least two has quotient value of one of the forms $\alpha q$, $q\beta$, or $q^2$.  For $m=1$ only the first two types occur, while all three occur for $m\ge2$.  Thus, over all $m$, the proper interval values lie in just three relative classes, represented by
\[
  [abaaa]_\theta,
  \qquad
  [baaab]_\theta,
  \qquad
  [baaabaaa]_\theta.
\]
Each is exact composite; finite set-product verification gives, for example,
\[
  [abaaa]_\theta=[abaa]_\theta\cdot[a]_\theta,
\]
\[
  [baaab]_\theta=[b]_\theta\cdot[aaab]_\theta,
\]
and
\[
  [baaabaaa]_\theta=[baaabaa]_\theta\cdot[a]_\theta.
\]
Thus no proper contiguous block has prime quotient product, and every rule in the displayed family is non-pleonastic.
\end{proof}

\begin{corollary}[UF--FRP separation]
\label{cor:uf-frp-separation}
Unique exact relative prime factorization does not imply FRP, even under finite relative quotient and finite relative prime spectrum.
\end{corollary}

\begin{proof}
Combine Theorem~\ref{thm:36-finite-structure} with the preceding infinite family of valid productions.
\end{proof}

The same family gives a particularly transparent separation between UF and tail saturation.  For every $m\ge1$,
\[
  \underbrace{P_q*\cdots*P_q}_{m\text{ copies}}*P_b
  =H:=[baaab]_\theta,
\]
where the notation on the left denotes quotient multiplication of the prime classes.  By global UF the class $H$ has the unique exact decomposition
\[
  \boxed{
  H=P_bP_r,
  \qquad P_r=[aaab]_\theta.
  }
\]
But the valid tails themselves under-saturate $H$:
\[
  \boxed{
  \overline{P_q^mP_b}\subsetneq H
  \qquad(m\ge1).
  }
\]
Consequently this one finite quotient proves simultaneously
\(\mathrm{UF}\not\Rightarrow\mathrm{VTE}_1, \qquad \mathrm{UF}\not\Rightarrow\mathrm{FRP}.\)
Since the relative quotient is finite, FSRP holds automatically; thus the witness has
\(m_{\rm ex}=1,\quad \rho_\partial=\infty,\quad \kappa_\partial<\infty.\)
It therefore separates exact-factor ambiguity from under-saturation defect in the strongest possible way: exact ambiguity is zero, while the presentation defect has unbounded length.

\begin{proposition}[Minimality among quotients of the witness]
The monoid $M_{36}$ has twenty monoid congruences.  Every proper
congruence quotient, with the induced observation on $L$, has FRP.  The
nineteen proper quotient sizes are
\(11,9,8,7,7,6,6,5,5,5,4,4,4,3,3,2,2,2,1.\)
Thus no proper quotient of this particular $36$-element observer continues to witness UF$+\neg$FRP.  This is a quotient-minimality statement for $M_{36}$, not a claim that $36$ is the absolute minimum among all finite monoids.
\end{proposition}

\begin{proof}
This is the exhaustive quotient audit in Appendix~\ref{app:36-computation}, item~\ref{check:quotient-minimality}.
\end{proof}

\section{A coupled nonregular FSRP--FRP separation}
\label{sec:wrapped}

The finite witness of Section~\ref{sec:36} proves that FSRP can compress an
infinite valid-rule family, but FSRP there is automatic from finiteness of the
relative quotient.  We now show that the same defect persists when the full
relative quotient is infinite and the target language is intrinsically
nonregular.

Let
\[
  \Gamma=\{a,b,c,d\},
  \qquad
  L^{\mathrm{wrap}}
  :=\{c^nwd^n:n\ge0,\ w\in\{a,b\}^+\}.
\]
Let $h_{36}:\{a,b\}^*\to M_{36}$ be the morphism from
Section~\ref{sec:36}.  Extend it to
$\widehat h_{36}:\Gamma^*\to M_{36}$ by
$\widehat h_{36}(c)=\widehat h_{36}(d)=1$.
Let $E=\{0,1\}^2$ with coordinatewise maximum as multiplication, and put
\[
  \chi(w)=
  \bigl([\,c\text{ occurs in }w\,],
        [\,d\text{ occurs in }w\,]\bigr)\in E.
\]
Then
\[
  h^{\mathrm{wrap}}(w):=(\widehat h_{36}(w),\chi(w))
  \colon \Gamma^*\to M_{36}\times E
\]
is a finite-monoid morphism.

For $s\in M_{36}\setminus\{1\}$ define the nonempty middle fiber
\[
  \mathsf{M}_s:=\{u\in\{a,b\}^+:h_{36}(u)=s\}.
\]
For $i,j\ge1$ and $r\in\mathbb Z$ put
\[
\begin{aligned}
  C_i&:=\{c^i\}, & D_j&:=\{d^j\},\\
  L_{i,s}&:=c^i\mathsf{M}_s,
    &R_{j,s}&:=\mathsf{M}_s d^j,\\
  W_{r,s}&:=\{c^iud^j:i,j\ge1,\ i-j=r,\ u\in\mathsf{M}_s\}.
\end{aligned}
\]
Write $C:=C_1$, $D:=D_1$, and $W_s:=W_{0,s}$.

\begin{lemma}[Wrapped live-class normal forms]
\label{lem:wrapped-normal-forms}
The live relative classes of
$(L^{\mathrm{wrap}},h^{\mathrm{wrap}})$ are exactly
\[
  [\eps],\quad C_i,\quad D_j,\quad \mathsf{M}_s,
  \quad L_{i,s},\quad R_{j,s},\quad W_{r,s},
\]
with the parameters above.  In particular, the relative quotient is infinite.
Moreover, the pair is unit-separated and $h^{\mathrm{wrap}}$-substitutable.
\end{lemma}

\begin{proof}
Every factor of a word $c^nwd^n$ has one of the forms
\[
  \eps,\quad c^i,\quad d^j,\quad
  u,\quad c^iu,\quad ud^j,\quad c^iud^j,
\]
where $u\in\{a,b\}^+$ and the wrapper exponents that occur are positive.
The two occurrence bits in $\chi$ separate pure middle, left-wrapper,
right-wrapper, and two-sided-wrapper factors; the first coordinate separates
pure wrapper factors from factors containing a nonempty middle word because
$h_{36}^{-1}(1)=\{\eps\}$.

All nonempty middle words have the same ordinary syntactic distribution in
$L^{\mathrm{wrap}}$: only their position inside the arbitrary nonempty middle
word matters, not their $a/b$ content.  Hence $h_{36}$ splits them precisely
into the classes $\mathsf M_s$.
For a left-wrapper factor $c^iu$, any accepting completion has the form
$c^p(c^iu)v d^q$ with $v\in\{a,b\}^*$ and necessarily $p+i=q$;
thus its distribution is determined by $i$, and $h_{36}(u)=s$ gives
$L_{i,s}$.  The right-wrapper case is dual.  If a factor already contains
both wrappers, an accepting completion has the form
$c^p(c^iud^j)d^q$, and membership is equivalent to
$p+i=q+j$.  Hence its distribution is determined by $i-j$, giving
$W_{r,s}$.

Conversely, the same equations distinguish different values of the displayed
parameters.  For example, $c^iad^i\in L^{\mathrm{wrap}}$ while
$c^jad^i\notin L^{\mathrm{wrap}}$ for $j\ne i$, so the classes $C_i$ are
pairwise distinct; the other parameters are separated similarly.  This also
shows that the relative quotient is infinite.

Unit separation follows because a nonempty word either contains $c$ or $d$,
in which case $\chi$ is nonzero, or lies in $\{a,b\}^+$, in which case its
$h_{36}$-value is nonidentity.  Finally, if two equal-$h^{\mathrm{wrap}}$-type
factors share an accepting context, the occurrence bits put them in the same
shape above.  The common completion equation forces equality of the exponent for pure
$c$- or pure $d$-factors, equality of $i$ in the left-wrapper case, equality
of $j$ in the right-wrapper case, and equality of $i-j$ in the two-sided
case; pure middle factors already have equal syntactic distribution.  Thus
they lie in the same displayed relative class, proving
$h^{\mathrm{wrap}}$-substitutability.
\end{proof}

\begin{lemma}[Middle-subsystem preservation]
\label{lem:wrapped-middle-subsystem}
For $u,v\in\{a,b\}^+$,
\[
  u\mathrel{\theta_{L^{\mathrm{wrap}},h^{\mathrm{wrap}}}}v
  \quad\Longleftrightarrow\quad
  h_{36}(u)=h_{36}(v).
\]
Consequently the live middle classes, their exact set products, the middle relative primes, and the correct and valid prime sequences contained wholly in the middle alphabet coincide with those of the $36$-element witness.
\end{lemma}

\begin{proof}
For middle words the wrapper-occurrence coordinates are both zero, so equality of $h^{\mathrm{wrap}}$-type is exactly equality of $h_{36}$-type.  Every two nonempty middle words have the same syntactic distribution in $L^{\mathrm{wrap}}$: a context can accept such a factor only by supplying equal numbers of outer $c$'s and $d$'s, independently of the internal $a/b$ word.  Thus the relative congruence restricted to $\{a,b\}^+$ is precisely $\ker h_{36}$.  Exact products and quotient products therefore agree with the finite witness, and primality, correctness, and pleonasticity are preserved under this identification.
\end{proof}

Let
\[
  \mathsf P_{36}
  :=\{s\in M_{36}\setminus\{1\}:\mathsf M_s
      \text{ is a relative prime}\}.
\]
By Theorem~\ref{thm:36-finite-structure}, $|\mathsf P_{36}|=15$.

\begin{theorem}[Prime spectrum and exact UF of the wrapped pair]
\label{thm:wrapped-primes-uf}
The relative prime spectrum of
$(L^{\mathrm{wrap}},h^{\mathrm{wrap}})$ is exactly
\[
  \boxed{
  \Pcal_{\mathrm{wrap}}
  =\{C,D\}
   \cup\{\mathsf M_p:p\in\mathsf P_{36}\}
   \cup\{W_s:s\in M_{36}\setminus\{1\}\}.
  }
\]
Hence $|\Pcal_{\mathrm{wrap}}|=2+15+35=52$.
Every live non-unit relative class has a unique exact prime factorization.
\end{theorem}

\begin{proof}
The normal forms satisfy the exact identities
\[
  C_i=C^i,
  \qquad
  D_j=D^j,
  \qquad
  L_{i,s}=C^i\mathsf M_s,
  \qquad
  R_{j,s}=\mathsf M_sD^j,
\]
and
\[
  W_{r,s}=
  \begin{cases}
    C^rW_s,&r>0,\\
    W_s,&r=0,\\
    W_sD^{-r},&r<0.
  \end{cases}
\]
Thus only $C,D$, the middle fibers, and the balanced wrapped classes $W_s$
can be prime.  By Lemma~\ref{lem:wrapped-middle-subsystem}, the prime middle fibers are exactly the fifteen $\mathsf M_p$ with $p\in\mathsf P_{36}$.

It remains to prove that every $W_s$ is prime.  Let $u_s$ be a shortest word
of $\mathsf M_s$.  Then $cu_sd$ is a shortest word of $W_s$.  By
Lemma~\ref{lem:shortest-word-cut}, any nontrivial exact binary factorization
would occur at a cut of this word.  The cut immediately after $c$ gives
\[
  C\cdot R_{1,s}=c\mathsf M_sd\subsetneq W_s,
\]
and the cut immediately before $d$ gives the dual strict inclusion.
A cut $u_s=xy$ inside the middle gives
\[
  L_{1,h_{36}(x)}\cdot R_{1,h_{36}(y)}
  \subseteq c\{a,b\}^+d
  \subsetneq W_s.
\]
All three types of products contain only words with one left and one right
wrapper symbol, whereas $c^2u_sd^2\in W_s$.  Hence no shortest-word cut is
exact, and $W_s$ is prime.

For uniqueness, the middle classes inherit the unique exact prime
factorizations of the $36$-element witness.  In a factorization of
$L_{i,s}$, no prime containing $d$ can occur, and a middle prime cannot
precede a $C$ without producing an $a/b$ symbol before a later $c$; therefore
the factorization is $C^i$ followed by the unique middle factorization of
$\mathsf M_s$.  The right-wrapper case is dual.

Finally consider $W_{r,s}$.  Its wrapper depth is unbounded inside one
relative class.  A product using only $C,D$, and middle primes has fixed
wrapper counts, so every exact factorization must contain a wrapped prime.
Two wrapped primes would place a $d$ before a later $c$, producing only dead words, so there is exactly one.  A middle prime cannot occur before that wrapped prime, since it would place an $a/b$ symbol before a later $c$; nor can one occur after it, since it would place an $a/b$ symbol after an earlier $d$.  Hence every exact factorization has the form
\(C^pW_tD^q\).  Exact equality with $W_{r,s}$ forces $t=s$ and, by comparing
the least possible left and right wrapper depths,
\[
  (p,q)=
  \begin{cases}
    (r,0),&r>0,\\
    (0,0),&r=0,\\
    (0,-r),&r<0.
  \end{cases}
\]
This proves global exact UF.
\end{proof}

We now describe the valid-rule languages directly.  Let
\[
  \mathcal A:=\{\mathsf M_p:p\in\mathsf P_{36}\}
\]
be the middle prime alphabet, and let
$\mu:\mathcal A^*\to M_{36}$ be the evaluation morphism.  Define
\[
  \Corr^{\mathrm{mid}}
  :=\{\beta\in\mathcal A^{\ge2}:\mu(\beta)\in\mathsf P_{36}\}
\]
and, for $s\ne1$,
\[
  \mathsf I_s
  :=\mu^{-1}(s)\cap\mathcal A^+
     \setminus
     \mathcal A^*\Corr^{\mathrm{mid}}\mathcal A^*.
\]
Both languages are regular, since $\mu$ has finite codomain.

\begin{theorem}[Intrinsic infinite-quotient FSRP--FRP separation]
\label{thm:wrapped-fsrp}
The pair $(L^{\mathrm{wrap}},h^{\mathrm{wrap}})$ is a nonregular
context-free, unit-separated, $h^{\mathrm{wrap}}$-substitutable pair with
infinite relative quotient and finite prime spectrum.  It has global exact UF
and satisfies FSRP but not FRP.  More precisely,
\[
  \Val_C=\Val_D=\varnothing,
\]
the valid-return language of each middle prime is exactly the corresponding
valid-return language of the $36$-element witness, and for every $s\ne1$,
\[
  \boxed{
  \Val_{W_s}=\{CW_sD\}\ \cup\ C\,\mathsf I_s\,D.
  }
\]
\end{theorem}

\begin{proof}
The grammar
\[
  S\to cSd\mid A,
  \qquad
  A\to aA\mid bA\mid a\mid b
\]
shows that $L^{\mathrm{wrap}}$ is context-free.  It is nonregular because
\[
  L^{\mathrm{wrap}}\cap c^*ad^*
  =\{c^nad^n:n\ge0\}.
\]
The remaining structural assertions up to finite prime spectrum and UF are
Lemma~\ref{lem:wrapped-normal-forms} and
Theorem~\ref{thm:wrapped-primes-uf}.

There is no correct branching rule with target $C$ or $D$, so their valid
languages are empty.  A correct sequence returning to a middle prime cannot contain $C,D$, or any $W_s$, because the $c/d$ occurrence bits would then be nonzero.  Lemma~\ref{lem:wrapped-middle-subsystem} therefore identifies its valid-return language with the corresponding language of the $36$-element witness, which is regular by Theorem~\ref{thm:finite-quotient-fsrp}.

Fix $s\ne1$.  A prime sequence whose quotient product is $W_s$ has exactly
one of two forms.  If it contains a wrapped prime, it can contain exactly one:
two would create a $dc$ boundary.  Shape then forces
\[
  C^kW_sD^k\qquad(k\ge1).
\]
If it contains no wrapped prime, shape forces
\[
  C^k\alpha D^k,
  \qquad
  k\ge1,
  \quad
  \alpha\in\mathcal A^+,
  \quad
  \mu(\alpha)=s.
\]
For $k\ge2$, either form is pleonastic because deleting one outer $C$ and one
outer $D$ leaves a proper contiguous block whose quotient class is the prime
$W_s$.  Thus validity forces $k=1$.

The word $CW_sD$ is valid: its two proper blocks of length two have quotient
classes $W_{1,s}$ and $W_{-1,s}$, both composite by
Theorem~\ref{thm:wrapped-primes-uf}.  For a sequence $C\alpha D$, any proper
block meeting $C$ but not $D$ has a left-wrapper quotient class and is
composite, and dually for a block meeting $D$ but not $C$.  Hence the only
possible prime-valued proper blocks lie entirely inside $\alpha$.  Such a
block is prime exactly when it belongs to $\Corr^{\mathrm{mid}}$.  Therefore
$C\alpha D$ is valid exactly when $\alpha\in\mathsf I_s$, proving the
displayed formula for $\Val_{W_s}$.  All valid languages are thus regular, so
FSRP holds.

Finally put
\[
  \mathsf M_a=[a],\qquad
  \mathsf M_q=[baaa],\qquad
  \mathsf M_b=[b],\qquad
  \mathsf M_{ab}=[ab]
\]
for the corresponding middle prime classes.  Quotient multiplication and
prime-valued proper intervals inside the middle subsystem are unchanged from
Section~\ref{sec:36}.  Hence
\[
  \mathsf M_{ab}
  \to
  \mathsf M_a\mathsf M_q^m\mathsf M_b
  \qquad(m\ge0)
\]
is valid for every $m$.  Thus FRP fails.  The same middle tails retain the
strict under-saturation from the finite witness, so VTE$_1$ fails as well.
\end{proof}

The theorem separates FRP from FSRP without appealing to finiteness of the
relative quotient.  The nonregular wrapper and the defective middle return
system are not disjoint components: the new wrapped primes $W_s$ have valid
languages obtained by inserting the finite-state middle irreducibles
$\mathsf I_s$ between the wrapper primes $C$ and $D$.

\begin{example}[A finite FRP language without UF]
\label{ex:nonuf-frp}
Let $L_f=\{abcd,apcd,bx\}$.  Take the three-element monoid $M=\{1,s,0\}$ with identity $1$, zero $0$, and $s^2=s$, and define $h(a)=h(c)=0$, $h(b)=s$, and $h(p)=h(d)=h(x)=1$.  Among distinct live factors of the same $h$-type, the only pairs sharing a context are $ab/ap$, $bc/pc$, $abc/apc$, $bcd/pcd$, and $abcd/apcd$; each pair has the same syntactic distribution.  Hence $L_f$ is $h$-substitutable.  The relative classes $X=\{abc,apc\}$, $Y=\{ab,ap\}$, $Z=\{c\}$, $A=\{a\}$, and $B=\{bc,pc\}$ satisfy the two distinct exact prime factorizations $X=Y\cdot Z=A\cdot B$.  By Corollary~\ref{cor:finite-candidates}, the shortest representatives of $A,Z,Y,B$ expose all possible nontrivial exact cuts; these cuts show that all four classes are prime.  Thus UF fails.  Since $L_f$ is finite, let $N$ be the maximum length of a live factor.  For any correct prime right-hand side of length $k$, concatenating shortest representatives yields a live factor of length at least $k$, so $k\le N$.  Over the finite prime alphabet there are therefore only finitely many valid rules, and FRP holds.
\end{example}

Corollary~\ref{cor:uf-frp-separation} and Example~\ref{ex:nonuf-frp} show that $\mathrm{UF}$ and $\mathrm{FRP}$ are incomparable.

\section{Disjoint-alphabet sums as an ambient comparison}
\label{sec:sum}

Let $(L_1,h_1)$ and $(L_2,h_2)$ be fixed-observation pairs over disjoint alphabets $\Sigma_1,\Sigma_2$.  Adjoin absorbing zeros $\bot_i$ to the monoids $M_i$, extend each morphism by sending letters of the other alphabet to $\bot_i$, and set
\(h=(\bar h_1,\bar h_2): (\Sigma_1\cup\Sigma_2)^* \to M_1^\bot\times M_2^\bot.\)
Put
\(L:=L_1\cup L_2.\)

\begin{theorem}[Disjoint-alphabet sum]
If each $(L_i,h_i)$ is $h_i$-substitutable and unit-separated, then $(L,h)$ is $h$-substitutable and unit-separated.  Its live relative classes and relative primes are the disjoint unions of the corresponding component classes and primes.  Moreover, for a component prime $P$,
\(\Val^{L}_P=\Val^{L_i}_P.\)
Consequently
\(\mathrm{FRP}(L,h) \iff \mathrm{FRP}(L_1,h_1)\wedge\mathrm{FRP}(L_2,h_2),\)
and similarly for FSRP.  Exact UF is also componentwise: the sum has UF iff each component has UF.
\end{theorem}

\begin{proof}
Every nonempty live factor lies in exactly one component alphabet; mixed-alphabet words are dead.  The absorbing coordinates of $h$ separate the two components, and the relative congruence restricted to $\Sigma_i^*$ is exactly that of $(L_i,h_i)$.  Thus live classes, $h$-substitutability, and unit separation are componentwise.

Any exact factorization of a live component class uses classes from that same component, so primality, exact factorization, and UF are componentwise.  The same separation applies to quotient returns: a prime sequence mixing components has a mixed representative and cannot return to a live prime.  Hence $\Val_P^L=\Val_P^{L_i}$ for each component prime $P$, proving the FRP and FSRP equivalences.
\end{proof}

As a comparison, take $L_0$ from Section~\ref{sec:L0} and an alphabet-renamed copy of the $36$-element witness.  Their disjoint sum
\(L^\ddagger:=L_0\cup\{c,d\}^+\)
is a nonregular context-free UF+FSRP+$\neg$FRP example by the theorem above.  This construction is only ambient: the nonregularity and the under-saturation defect remain in separate components.  The wrapped language of Section~\ref{sec:wrapped} is stronger, since its infinite quotient and finite-state defective return dynamics occur in one coupled relative-prime system.

\section{Strong structural reconstruction from positive data}
\label{sec:strong-learning}

The strong PTLD theorem below is stated for $L\subseteq\Sigma^+$ so that the learner need not carry a separate empty-word start case.  The relative-prime and presentation theory developed in the preceding sections is formulated over $\Sigma^*$ and does not rely on this notational restriction.

The structural results above suggest two distinct learning targets.  Under PTLD the canonical object is the finite direct relative-prime grammar.  Under the weaker FSRP condition, the canonical object is instead the finite-state residual controller for the valid-rule languages.  This section gives a direct positive-data strong learner for the PTLD branch and a limit-canonicalization theorem for the FSRP branch.

We use the strong Gold viewpoint of Clark~\cite{Clark2013Trees}: a learner succeeds strongly when its hypotheses eventually stabilize to a canonical representation of the target, not merely to an arbitrary grammar generating the same language.  In the PTLD branch, each update is polynomial in the current sample size, and the target admits an explicit finite characteristic sample.  Quantitative word-length bounds below are stated using a finite cut-separation parameter in addition to canonical size and thickness.

\begin{definition}[Positive text]
A \emph{positive text} for $L$ is an infinite sequence $T=(w_0,w_1,\ldots)$ of elements of $L$ in which every word of $L$ occurs at least once.  Write $T[n]$ for its length-$n$ prefix and $\operatorname{content}(T[n])$ for the finite set of words observed so far.
\end{definition}

\begin{definition}[Weak behavioral correctness]
A CFG-valued learner $\mathcal W$ is \emph{weakly behaviorally correct} on $L$ if for every positive text $T$ for $L$ there is $N$ such that
\(L(\mathcal W(T[n]))=L \qquad(n\ge N).\)
The hypotheses themselves need not stabilize.
\end{definition}

\begin{definition}[Strong Gold identification]
Let $\Can(L,h)$ be a fixed canonical representation of the pair $(L,h)$.  A learner $\mathcal A$ \emph{strongly identifies} $\Can(L,h)$ in the Gold sense if for every positive text $T$ for $L$ there is $N$ such that
\(\mathcal A(T[n])=\Can(L,h) \qquad(n\ge N),\)
where equality may be replaced by isomorphism when canonical names have not been fixed.
\end{definition}

\begin{definition}[Characteristic sample]
For a learner $\mathcal A$ and a canonical target $R=\Can(L,h)$, a finite set $\CS\subseteq L$ is a \emph{characteristic sample} if every finite sample $D$ with
\(\CS\subseteq D\subseteq L\)
forces $\mathcal A(D)=R$.
For finite $D$, write
\(\|D\|:=\sum_{w\in D}|w|.\)
Polynomial update time is measured as a polynomial in $\|D\|$.
\end{definition}

\subsection{The canonical PTLD target}

Assume throughout this subsection that $L$ is $h$-substitutable, unit-separated, has finitely many relative primes, and satisfies PTLD.  Let
\(\Pcal=\{P_1,\ldots,P_q\}.\)
By Theorems~\ref{thm:ptld-uf} and~\ref{thm:ptld-vte-vtd}, exact factorization is unique, VTE and VTD hold, and $|\Vcal|\le q^2$.

Let $G^\star_{L,h}$ be the direct relative-prime grammar whose prime nonterminals are the live relative primes, whose branching rules are exactly the valid productions, whose lexical rules are $P\to a$ whenever $a\in P\cap\Sigma$, and whose start rules use the unique exact prime factorization of each accepting relative class.  We canonically name every prime by its shortlex least word
\(\omega(P):=\min_{\mathrm{shortlex}}P.\)

For any sentential form $\alpha\in(\Pcal\cup\Sigma)^*$ of $G^\star_{L,h}$, let
\(\ell(\alpha):=\min\{|w|:\alpha\Rightarrow^*w\}.\)
Define the \emph{structural thickness}
\(\tau(G^\star):= \max_{A\to\alpha\in G^\star}\ell(\alpha).\)
Write simply $\tau$ when the target is clear.  Since every prime is non-unit, $|\omega(P)|\le\tau$ for each prime appearing on a right-hand side or start factorization.

\begin{proposition}[Trimness of the canonical PTLD grammar]
\label{prop:trim}
Every live relative prime is reachable and productive in $G^\star_{L,h}$; hence $G^\star_{L,h}$ is trim after deletion of vacuous start structure.
\end{proposition}

\begin{proof}
Productivity is Theorem~\ref{thm:direct-presentation}.  For reachability, choose $x\in P$ and a context $\ell x r\in L$.  Factor the live classes $[\ell]$ and $[r]$ exactly into primes and concatenate those factors with $P$.  This prime sequence has quotient equal to the accepting class $[\ell x r]$ and therefore has setwise product contained in that class.  By Theorem~\ref{thm:prime-refinement}, the exact start factorization of $[\ell x r]$ derives this sequence.  Expanding the left and right factors to $\ell$ and $r$ yields a derivation $S\Rightarrow^*\ell P r$.
\end{proof}

\begin{lemma}[Bounded canonical context]
\label{lem:bounded-context}
For every prime $P$ there is a terminal context $(\ell_P,r_P)$ such that
\(S\Rightarrow^*\ell_P P r_P\)
and
\(|\ell_P|+|r_P|\le q\tau.\)
\end{lemma}

\begin{proof}
Choose a terminal context of minimum total length and inspect a root-to-$P$ path in a corresponding partial derivation.  If a prime nonterminal occurred twice on that path, stop at the earlier occurrence and expand all off-path siblings to shortest terminal yields.  Every branching rule has at least two non-unit prime symbols, and every non-unit prime has shortest yield at least one, so removing the repeated spine segment removes a positive terminal contribution and produces a strictly shorter context around the same prime.  Hence the prime spine is simple and contains at most $q$ prime nodes.  Including the start step, there are at most $q$ off-spine contributions, each bounded by $\tau$ by the definition of thickness.
\end{proof}

\subsection{An h-guarded substring learner}

For a finite positive sample $D\subseteq L$, let $\operatorname{Sub}(D)$ be the set of nonempty substrings of words in $D$.  Define a finite CFG $W_h(D)$ with a fresh start symbol $S_D$, nonterminals $\langle u\rangle$ for $u\in\operatorname{Sub}(D)$, start rules
\(S_D\to\langle w\rangle\)
for every $w\in D$, and rules
\(\langle uv\rangle\to\langle u\rangle\langle v\rangle,\)
whenever $u,v,uv\in\operatorname{Sub}(D)$,
\(\langle a\rangle\to a\)
for observed letters, and symmetric substitution rules
\(\langle u\rangle\leftrightarrow\langle v\rangle\)
whenever
\(h(u)=h(v)\)
and there is a sample context $(\ell,r)$ with
\(\ell ur,\ell vr\in D.\)
Here $\leftrightarrow$ abbreviates the two unit productions in opposite directions.  This is the ordinary substring/substitution construction guarded by equality of the fixed observer type.

\begin{theorem}[Guarded weak soundness]
\label{thm:weak-sound}
For every positive sample $D\subseteq L$ and every $u\in\operatorname{Sub}(D)$,
\(L(W_h(D),\langle u\rangle)\subseteq[u]_{L,h}.\)
Consequently
\(L(W_h(D))\subseteq L.\)
\end{theorem}

\begin{proof}
Map $\langle u\rangle$ to the true relative class $[u]$.  Binary rules are sound because $\theta_{L,h}$ is a congruence.  Every guarded substitution preserves the true relative class by Lemma~\ref{lem:type-context-collapse}.  Induction on derivations proves the first inclusion.  Every start class contains a sample word in $L$ and is therefore an accepting relative class, which proves the second.
\end{proof}

The preceding theorem gives a useful one-sided invariant: the learner may undergenerate before its characteristic data arrive, but it never creates a string outside the target language.

\subsection{Finite signatures and exact observed relative classes}

For each $u\in\operatorname{Sub}(D)$ fix one observed occurrence
\(c_D(u)=(\ell_u,r_u), \qquad \ell_u u r_u\in D.\)
For any CFG $W$ define the finite signature
\[
  \sigma_{D,W}(u)
  :=
  \left(
    h(u),
    \bigl(
      {\bf1}[\ell_xur_x\in L(W)]
    \bigr)_{x\in\operatorname{Sub}(D)}
  \right).
\]

\begin{lemma}[No false relative merges]
\label{lem:no-false-merge}
For $W=W_h(D)$,
\(\sigma_{D,W}(u)=\sigma_{D,W}(v) \Longrightarrow u\theta_{L,h}v.\)
If in addition $L(W)=L$, then for all observed substrings
\(\sigma_{D,W}(u)=\sigma_{D,W}(v) \iff u\theta_{L,h}v.\)
\end{lemma}

\begin{proof}
Equality of signatures gives $h(u)=h(v)$.  The coordinate indexed by $u$ satisfies $\ell_u u r_u\in D\subseteq L(W)$.  Hence signature equality implies $\ell_u v r_u\in L(W)\subseteq L$ by Theorem~\ref{thm:weak-sound}.  Thus $u$ and $v$ have equal observer type and a common $L$-context, so Lemma~\ref{lem:type-context-collapse} gives $u\theta v$.  If $L(W)=L$, the converse follows directly from equality of full syntactic distributions and observer types.
\end{proof}

Thus, after the weak core becomes complete, the true restriction of $\theta_{L,h}$ to all observed substrings is computable by finitely many ordinary CFG membership tests.

\subsection{Finite certificates for relative primality}

For a prime $P$ and a nontrivial cut
\[
  \omega(P)=uv,
  \qquad u,v\ne\eps,
\]
put
\[
  Y_{P,u}:=[u]_\theta,
  \qquad
  Z_{P,v}:=[v]_\theta.
\]
Since $\omega(P)\in P$, congruence gives
\[
  Y_{P,u}\cdot Z_{P,v}\subseteq P.
\]
Both factors are live and non-unit.  Hence primality of $P$ implies
\[
  Y_{P,u}\cdot Z_{P,v}\subsetneq P.
\]

Define the \emph{prime cut-separation radius}
\[
  \eta_{\rm cut}(L,h)
  :=
  \max
  \left\{
    \min\{|x|:
      x\in P\setminus
      (Y_{P,u}\cdot Z_{P,v})
    \}:
    \begin{array}{l}
      P\in\Pcal,\\
      \omega(P)=uv,\ u,v\ne\eps
    \end{array}
  \right\},
\]
with $\eta_{\rm cut}(L,h)=0$ if no canonical prime word has a nontrivial cut.

Under PTLD the set difference in this definition has a more transparent
prefix interpretation.

\begin{lemma}[Cut separation is prime-prefix escape]
\label{lem:cut-prefix-escape}
Assume that $L$ is $h$-substitutable, unit-separated, and satisfies PTLD.  Let $P$ be a relative prime and let
$\omega(P)=uv$ be a nontrivial cut.  Put
$Y=[u]_\theta$ and $Z=[v]_\theta$.  Then
\[
  P\cap Y\Sigma^*=Y\cdot Z.
\]
Consequently
\[
  \min\{|x|:x\in P\setminus(Y\cdot Z)\}
  =
  \min\{|x|:x\in P,\ x\notin Y\Sigma^*\}.
\]
Thus, on the PTLD branch, $\eta_{\rm cut}$ is exactly a quantitative
prime-prefix escape radius for the finitely many cuts of the canonical
prime words.
\end{lemma}

\begin{proof}
The inclusion $Y\cdot Z\subseteq P$ follows from congruence and
$uv\in P$.  Conversely, let $x\in P\cap Y\Sigma^*$, say $x=yr$ with
$y\in Y$.  Since $u\theta y$ and both $uv$ and $yr$ lie in the prime
class $P$, Lemma~\ref{lem:ptld-collapse} gives $v\theta r$.  Hence
$r\in Z$ and $x\in Y\cdot Z$.  This proves the equality, and the
minimum-length identity is immediate.
\end{proof}

\begin{lemma}[Finite prime cut separation]
\label{lem:finite-cut-separation}
Assume finite relative prime spectrum and unit separation.  Then
$\eta_{\rm cut}(L,h)<\infty$.  For every prime $P$ and every nontrivial cut
$\omega(P)=uv$, there is a word
\[
  x_{P,u,v}\in
  P\setminus([u]_\theta\cdot[v]_\theta)
\]
with
\[
  |x_{P,u,v}|\le\eta_{\rm cut}(L,h).
\]
Consequently, if $(\ell_P,r_P)$ is a live context for $P$, then
\[
  \ell_Px_{P,u,v}r_P\in L
\]
is a positive certificate refuting the false exact split
$P=[u]_\theta\cdot[v]_\theta$.
\end{lemma}

\begin{proof}
For every displayed cut, $[u]_\theta$ and $[v]_\theta$ are live non-unit relative classes.  If
\[
  P=[u]_\theta\cdot[v]_\theta,
\]
then $P$ would be composite, contrary to primality.  Thus the displayed set difference is nonempty.  There are finitely many primes and finitely many cuts of each canonical prime word, so the maximum of the finitely many minimum witness lengths is finite.
\end{proof}

\begin{corollary}[Finite prime certification]
\label{cor:prime-certification}
Suppose the observed relative partition is exact.  If, for every true prime $P$, the sample contains $\omega(P)$ and the cut-separation witness $\ell_Px_{P,u,v}r_P$ for every nontrivial cut $\omega(P)=uv$, then the cut test of Corollary~\ref{cor:finite-candidates} identifies exactly the true relative primes among the observed classes.
\end{corollary}

\begin{proof}
For a true prime, every candidate split exposed by a cut of $\omega(P)$ is refuted by its corresponding cut-separation witness.  Conversely, if a live class $X$ is composite, write $X=Y\cdot Z$ exactly with $Y,Z$ live non-unit classes.  For every $w\in X$, exact equality gives a factorization $w=yz$ with $y\in Y$ and $z\in Z$.  Hence the corresponding cut candidate is genuine and no positive word of $X$ can refute it.  Therefore a composite class is never certified as prime.
\end{proof}

\subsection{Polynomial recovery of exact factors and valid rules}

Once the observed partition and prime set are correct, Corollary~\ref{cor:prime-dag} turns factor recovery into a finite shortest-path problem.  For an observed word $w=a_1\cdots a_n$, form its \emph{relative prime interval DAG} with vertices $0,\ldots,n$ and an edge $i\to j$ whenever the substring $a_{i+1}\cdots a_j$ belongs to an observed prime class.  A minimum-edge path reads the exact prime factorization of $[w]$.

To recover valid productions, enumerate triples $(N,A,Q)$ of an observed target prime $N$, a first prime $A$, and an observed live relative class $Q$.  Let
\(\alpha(Q)=\operatorname{pf}(Q)\)
be the exact prime factorization returned by the interval DAG.  Test the candidate
\(N\to A\alpha(Q)\)
for quotient correctness and reject it if any proper contiguous block of length at least two has prime quotient.  All tests reduce to the finite observed partition and CFG membership in $W_h(D)$.

\begin{lemma}[Rule-recovery completeness under PTLD]
\label{lem:rule-recovery}
After the observed partition and prime set are correct, the preceding triple procedure returns all and only the valid productions whose witnesses occur in the sample.  If every target valid production has its canonical witness in the sample, the recovered branching-rule set is exactly $\Vcal$.
\end{lemma}

\begin{proof}
Soundness is by the explicit correctness and non-pleonasticity tests.  For completeness, let $N\to A\beta$ be target-valid.  VTE makes $\beta$ the exact prime factorization of its tail class $Q=[\overline\beta]$.  A canonical rule witness exposes a representative of $Q$ as a substring.  Corollary~\ref{cor:prime-dag} therefore returns $\alpha(Q)=\beta$, so the triple $(N,A,Q)$ regenerates the target rule.
\end{proof}

\subsection{Characteristic sample}

For every prime $P$, among the contexts of minimum total length satisfying Lemma~\ref{lem:bounded-context}, fix the lexicographically least pair $(\ell_P,r_P)$.  This tie-break makes the learner deterministic; once the relative partition is exact, any live context for $P$ induces the same class test.  Define four finite families.

\paragraph{Anchors.}
\(\CS_{\rm anc} := \{\ell_P\omega(P)r_P:P\in\Pcal\}.\)

\paragraph{Rule witnesses.}
For every valid rule $P\to\alpha$,
\(\CS_{\rm rule} := \{\ell_P\omega(\alpha)r_P:P\to\alpha\in\Vcal\},\)
where $\omega(\alpha)$ concatenates the canonical shortest representatives of the prime symbols in $\alpha$.

\paragraph{Cut-separation witnesses.}
\(\CS_{\rm cut} := \{\ell_Px_{P,u,v}r_P: P\in\Pcal, \omega(P)=uv, u,v\ne\eps\}.\)

\paragraph{Lexical and start witnesses.}
For every live terminal $a$, include one sentence $\ell_{[a]}a r_{[a]}$; this ensures recovery of all lexical alternatives even when several letters lie in the same relative prime.  For each accepting relative class $X$, include the canonical word obtained by concatenating the $\omega(P)$ along its unique exact prime factorization.  Let the union of these two finite families be $\CS_{\rm lex/start}$.

Set
\[
  \CS^{\rm PTLD}(L,h)
  :=
  \CS_{\rm anc}
  \cup\CS_{\rm rule}
  \cup\CS_{\rm cut}
  \cup\CS_{\rm lex/start}.
\]

\begin{lemma}[Characteristic weak completeness]
\label{lem:weak-core}
If $D\subseteq L$ contains $\CS_{\rm anc}\cup\CS_{\rm rule}\cup\CS_{\rm lex/start}$, then
\(L(W_h(D))=L.\)
\end{lemma}

\begin{proof}
Soundness is Theorem~\ref{thm:weak-sound}.  For the converse, compare an anchor $\ell_P\omega(P)r_P$ with the rule witness $\ell_P\omega(\alpha)r_P$ for a valid rule $P\to\alpha$.  Their central factors have the same $h$-type and the same sample context, so $W_h(D)$ contains the guarded substitution between their substring nonterminals.  Binary splitting then simulates the canonical production $P\to\alpha$.  Lexical witnesses provide every terminal alternative.  Finally, each start witness is itself a positive sentence whose substring nonterminal splits into the unique exact prime factorization of its accepting class.  Hence $W_h(D)$ simulates every derivation of $G^\star_{L,h}$ and generates all of $L$.
\end{proof}

The size of the characteristic set is explicit.  PTLD gives $|\Vcal|\le q^2$, every canonical prime word has length at most $\tau$, and there are at most $|M|$ accepting relative classes.  There are at most
\[
  \sum_{P\in\Pcal}(|\omega(P)|-1)\le q(\tau-1)
\]
cut-separation witnesses.  Therefore
\[
  |\CS^{\rm PTLD}| \le q^2+q\tau+|\Sigma|+|M|.
\]
Anchors and rule witnesses have length at most $(q+1)\tau$, a cut-separation witness has length at most $q\tau+\eta_{\rm cut}(L,h)$, lexical witnesses have length at most $q\tau+1$, and start witnesses have length at most $\tau$.  Consequently
\[
  \max_{w\in\CS^{\rm PTLD}}|w|
  \le
  \max\bigl\{(q+1)\tau,\;q\tau+\eta_{\rm cut}(L,h),\;q\tau+1\bigr\}.
\]
Thus the target has an explicit finite characteristic sample.  Quantitatively, for fixed $h$ and fixed alphabet, its total data size is polynomial in $q$, $\tau$, and $\eta_{\rm cut}(L,h)$.

\subsection{Relative-ASGOLD}

The name follows Clark's ASGOLD strong learner~\cite{Clark2013Trees}.  Relative-ASGOLD is its fixed-$h$ guarded analogue: the observable partition is refined by the fixed monoid type, while the reconstruction target is the canonical relative-prime grammar.

\paragraph{Primitive tests.}
Once the weak core is complete and the canonical context $(\ell_P,r_P)$ of an observed prime $P$ is present, membership of an arbitrary word $z$ in $P$ is tested by
\[
  \operatorname{ClassTest}_D(z,P)
  \iff
  h(z)=h(\omega(P))
  \ \text{and}\ 
  \ell_P z r_P\in L(W_h(D)).
\]
For an observed candidate class $X$ with canonical representative $w$, the prime test is finite.  For each nontrivial cut $w=uv$, let $Y=[u]$ and $Z=[v]$ in the observed partition.  An observed word $x\in X$ refutes the candidate exact split $X=Y\cdot Z$ exactly when no cut $x=x_1x_2$ satisfies both $\operatorname{ClassTest}_D(x_1,Y)$ and $\operatorname{ClassTest}_D(x_2,Z)$.  All such cuts are among observed substrings, so this is a finite search; $\operatorname{IsPrime}_D(X)$ accepts precisely when every canonical cut has an observed refuter.  The cut-separation part of the characteristic sample guarantees those refuters for every true prime.

Using the same primitive, the canonicalizer implements $\operatorname{PrimeFact}_D$ by a minimum path in the relative prime interval DAG, $\operatorname{Correct}_D(P,\alpha)$ by testing the quotient class of a representative of $\alpha$, and $\operatorname{Valid}_D(P,\alpha)$ by additionally testing all proper contiguous blocks of length at least two.  Thus every structural test used below reduces to finitely many $h$-computations and CFG-membership tests.

The strong learner takes a finite positive sample $D$ and performs the following finite operations:
\begin{enumerate}[leftmargin=2.2em]
\item construct the guarded weak grammar $W_h(D)$;
\item partition observed substrings by the signatures $\sigma_{D,W_h(D)}$;
\item choose the shortlex least observed representative of each class;
\item apply the finite observed-refuter cut test above, justified by Corollary~\ref{cor:prime-certification}, to select sample primes;
\item recover prime decompositions by minimum paths in the relative prime interval DAGs;
\item enumerate the $(N,A,Q)$ triples and keep exactly the correct, non-pleonastic rules;
\item recover observed lexical rules and accepting start factorizations;
\item output the resulting grammar using canonical shortlex labels.
\end{enumerate}
Call this learner \emph{Relative-ASGOLD}.

\begin{theorem}[PTLD strong structural reconstruction]
\label{thm:ptld-strong}
Fix a finite monoid morphism $h:\Sigma^*\to M$.  Let $L\subseteq\Sigma^+$ be an $h$-substitutable context-free language satisfying unit separation, finite relative prime spectrum, and PTLD.  Then Relative-ASGOLD identifies the canonical direct relative-prime grammar $G^\star_{L,h}$ strongly from positive strings alone.

Moreover, each hypothesis update is computable in time polynomial in the size of the current finite sample.  The target admits an explicit finite characteristic sample; its cardinality is at most $q^2+q\tau+|\Sigma|+|M|$, and its maximum word length is bounded by
\[
  \max\bigl\{(q+1)\tau,\;q\tau+\eta_{\rm cut}(L,h),\;q\tau+1\bigr\}.
\]
\end{theorem}

\begin{proof}
Let $D$ contain $\CS^{\rm PTLD}$.  Lemma~\ref{lem:weak-core} gives $L(W_h(D))=L$, so Lemma~\ref{lem:no-false-merge} makes the observed relative partition exact.  Corollary~\ref{cor:prime-certification} then recovers exactly the true observed primes.  All target primes are present because their anchors are present.  Corollary~\ref{cor:prime-dag} recovers exact prime factorizations, and Lemma~\ref{lem:rule-recovery} recovers exactly the valid rules.  Lexical and start witnesses recover the remaining canonical structure.  Canonical shortlex names are already present and cannot later change, so every positive supersample of $D$ produces literally the same grammar $G^\star_{L,h}$.

For complexity, let $N=\|D\|$ and put $s=|\operatorname{Sub}(D)|$.  Then $s\le N(N+1)/2=O(N^2)$.  The grammar $W_h(D)$ has $s$ substring nonterminals; its binary split rules are bounded by $O(sN)$, while its guarded substitution rules are bounded by $O(s^2)$, so its total size is polynomial in $N$.  Signature construction uses at most $s^2=O(N^4)$ CFG-membership tests, each on a word of length $O(N)$ in a grammar of polynomial size.  The observed-refuter prime tests inspect at most $O(sN)$ canonical cuts and only polynomially many observed cuts and class tests.  Every relative prime interval DAG has $O(N)$ vertices and $O(N^2)$ possible edges, so minimum-path factor recovery is polynomial.  Finally, the rule stage considers at most $s^3=O(N^6)$ triples $(N,A,Q)$ of observed classes, with correctness and non-pleonasticity checked by polynomially many class and interval tests.  Since ordinary CFG membership is polynomial in grammar and input size, every stage of one hypothesis update is polynomial in $N$.  The characteristic-sample bounds were established above.
\end{proof}

\begin{remark}[Relation to the fixed-$h$ weak learner]
Theorem~\ref{thm:ptld-strong} does not require the general fixed-$h$ weak learner as a black box: the guarded substring grammar supplies the weak core directly on the PTLD subclass.  The earlier fixed-$h$ reconstruction theorem remains more general, while the present theorem trades the algebraic PTLD restriction for convergence to a canonical structural representation.
\end{remark}

\subsection{A polynomial-data PTLD subbranch and the thickness obstruction}

The only uncontrolled length parameter in Theorem~\ref{thm:ptld-strong}
is the cut-separation radius.  There is, however, a natural PTLD subbranch
on which it vanishes identically.

\begin{definition}[Lexically anchored PTLD]
A PTLD pair with finite relative prime spectrum is \emph{lexically
anchored} if every relative prime $P$ contains at least one terminal
letter, that is,
\[
  P\cap\Sigma\ne\varnothing
  \qquad(P\in\Pcal).
\]
\end{definition}

\begin{corollary}[Polynomial time and data on the lexically anchored branch]
\label{cor:lexically-anchored-polynomial}
Fix the finite alphabet $\Sigma$ and finite observer $h$.  On the class of
unit-separated, $h$-substitutable, lexically anchored PTLD languages with
finite relative prime spectrum, Relative-ASGOLD has polynomial update time
and polynomial characteristic data in the size of the explicit canonical
direct relative-prime grammar.
\end{corollary}

\begin{proof}
If every prime contains a letter, then its canonical shortlex word has
length one.  Hence no canonical prime word has a nontrivial cut and
\[
  \eta_{\rm cut}(L,h)=0.
\]
Moreover, for a prime sequence $\alpha=P_1\cdots P_m$ one has
$\ell(\alpha)=m$, so the structural thickness $\tau$ is just the largest
right-hand-side length of the explicit canonical grammar.  Theorem~\ref{thm:ptld-strong}
therefore gives
\[
  |\CS^{\rm PTLD}|\le q^2+q\tau+|\Sigma|+|M|
\]
and
\[
  \max_{w\in\CS^{\rm PTLD}}|w|\le(q+1)\tau.
\]
For fixed $\Sigma$ and $h$, both $q$ and $\tau$ are bounded by the explicit
encoding size of the canonical grammar.  Thus the total characteristic
data size is polynomial in that representation size.  Polynomial update
time is already part of Theorem~\ref{thm:ptld-strong}.
\end{proof}

This does not imply a polynomial transfer from an arbitrary compact CFG
presentation to the canonical PTLD representation.  The obstruction is
already present for finite unary languages.

\begin{proposition}[Exponential thickness under compact CFG presentation]
\label{prop:compact-thickness-obstruction}
There is a family of unit-separated substitutable PTLD languages having
$q=1$ and $\eta_{\rm cut}=0$ that admit CFGs of size $O(n)$ but whose
canonical PTLD thickness is $2^n$.
\end{proposition}

\begin{proof}
Let
\[
  L_n=\{a^{2^n}\}
\]
over the unary alphabet, with the trivial one-element group observer.
The live factors are $a^i$ for $0\le i\le2^n$.  If $a^i$ and $a^j$
share an accepting context $(a^r,a^s)$, then
$r+i+s=2^n=r+j+s$, so $i=j$; hence $L_n$ is substitutable.  The empty
class is the singleton $\{\eps\}$, and PTLD is automatic for the trivial
group.

Every live non-unit relative class is the singleton $\{a^i\}$, and
\[
  \{a^i\}=\{a\}^i.
\]
Thus the only relative prime is $P=\{a\}$, so $q=1$ and
$\eta_{\rm cut}=0$.  The accepting class has exact start factorization
$P^{2^n}$, whence the canonical direct grammar has
\[
  \tau=2^n.
\]
On the other hand $L_n$ has an acyclic doubling grammar of size $O(n)$:
\[
  A_0\to a,\qquad
  A_{i+1}\to A_iA_i\quad(0\le i<n),\qquad
  S\to A_n.
\]
Therefore no polynomial in the size of an arbitrary compact generating
CFG can bound the canonical PTLD thickness in general.
\end{proof}

The proposition isolates the remaining representation issue: the general
PTLD theorem is parameterized by canonical thickness, just as compact CFGs
can hide exponentially long shortest yields.  Polynomial transfer from an
arbitrary generating grammar therefore requires additional restrictions on
that presentation or on the language class.

\subsection{A PTLD language beyond every finite \texorpdfstring{$k,\ell$}{k,l}-substitutable class}

\begin{example}[A nonregular PTLD language beyond all finite context windows]
\label{ex:Lstar}
Over $\Sigma=\{a,b,x,y,z\}$ put
\[
L^\star
=
\{a^nxb^n:n\ge0\}
\ \cup\
\{a^nyb^n:n\ge0\}
\ \cup\
\{za^nxb^n:n\ge0\}.
\]
Let the observer take values in the finite group $C_2\times C_2$ and be defined on letters by
\[
h(a)=h(b)=h(x)=(0,0),\qquad
h(y)=(1,0),\qquad
h(z)=(0,1).
\]
The language is deterministic context-free: a pushdown automaton optionally records the initial $z$, pushes the $a$-run, reads the unique center $x$ or $y$, and pops on the $b$-run, rejecting the $y$-branch when the initial marker was $z$.  It is nonregular because
\[
L^\star\cap a^*xb^*=\{a^nxb^n:n\ge0\}.
\]

We first compare the example directly with Yoshinaka's bounded-context hierarchy.  Recall that $L$ is $k,\ell$-substitutable when, for every $u\in\Sigma^k$ and $v\in\Sigma^\ell$, two nonempty guarded strings $uy_1v$ and $uy_2v$ that occur in one common outer context are interchangeable in every outer context~\cite{Yoshinaka2008}.

\begin{proposition}[Outside every finite $k,\ell$ window]
\label{prop:Lstar-not-kl}
For every finite $k,\ell\ge0$, the language $L^\star$ is not $k,\ell$-substitutable.
\end{proposition}

\begin{proof}
Fix $k,\ell$ and choose $n\ge\max\{k,\ell\}$.  Put
\[
u=a^k,\qquad v=b^\ell,
\]
\[
y_1=a^{n-k}xb^{n-\ell},\qquad
y_2=a^{n-k}yb^{n-\ell}.
\]
Then
\[
uy_1v=a^nxb^n\in L^\star,
\qquad
uy_2v=a^nyb^n\in L^\star,
\]
so the two guarded strings occur in the common empty outer context.  But the outer left context $z$ separates them:
\[
zuy_1v=za^nxb^n\in L^\star,
\qquad
zuy_2v=za^nyb^n\notin L^\star.
\]
Hence the $k,\ell$ substitution condition fails.  Since $k,\ell$ were arbitrary, one language lies outside every finite level of the hierarchy.
\end{proof}

The live syntactic classes admit simple normal forms.  Put
\[
U=\{\eps\},\qquad
A_i=\{a^i\}\ (i\ge1),\qquad
B_j=\{b^j\}\ (j\ge1),
\]
\[
H_i=\{za^i\}\ (i\ge0),
\]
and, for $d\in\mathbb Z$,
\[
X_d=\{a^ixb^j:i,j\ge0,\ i-j=d\},
\qquad
Y_d=\{a^iyb^j:i,j\ge0,\ i-j=d\}.
\]
Finally, for $d\ge0$, put
\[
T_d=\{za^{j+d}xb^j:j\ge0\}.
\]
Every live factor is in exactly one displayed class.  Indeed, a factor either lies wholly in an $a$-run, wholly in a $b$-run, is a prefix beginning with the unique $z$, or contains the unique center $x$ or $y$; a factor containing both $z$ and $x$ must be a prefix through the center, which gives $T_d$ with $d\ge0$.

The displayed sets are precisely the syntactic classes.  For $X_d$, an accepting context has the form
\[
(a^p,b^{p+d})\quad\text{or}\quad(za^p,b^{p+d}),
\qquad p\ge\max\{0,-d\},
\]
whereas for $Y_d$ only the first family is available.  Thus each of these classes depends only on $d$.  For $T_d$ the two-sided distribution is the singleton
\[
\D_{L^\star}(T_d)=\{(\eps,b^d)\}.
\]
The singleton families $A_i,B_j,H_i$ are distinguished by their required completion lengths.  These descriptions also separate the displayed families from one another.

They make $h$-substitutability transparent.  The observer types are
\[
h(U)=h(A_i)=h(B_j)=h(X_d)=(0,0),
\]
\[
h(Y_d)=(1,0),
\qquad
h(H_i)=h(T_d)=(0,1).
\]
There is no live class of type $(1,1)$.  Among classes of type $(0,0)$, distinct $A_i$'s or $B_j$'s require different completion lengths; an $A_i$-context has the unique center to the right of the factor, a $B_j$-context has it to the left, and an $X_d$-context has the center inside the factor, so these three shapes cannot share an accepting context across families.  The empty class $U$ cannot share an accepting context with a nonempty class of the same type: inserting $A_i$ or $B_j$ destroys the balance of a completed sentence, while inserting an $X_d$ introduces a second center.  Among the $Y_d$'s a common accepting context forces the same imbalance $d$.  Among type $(0,1)$ classes, an $H_i$ requires the center in its right context while a $T_d$ already contains the center; hence the two families have disjoint distributions, and within each family the completion length fixes the index.  Therefore equal $h$-type plus one common accepting context always forces equality of the syntactic class.  Thus $L^\star$ is $h$-substitutable.  The same normal forms give $[\eps]_{L^\star,h}=U=\{\eps\}$, so unit separation holds.

Because $C_2\times C_2$ is a group, PTLD holds: from $qr=p=qr'$ one cancels $q$ on the left and obtains $r=r'$.  The relative prime spectrum has exactly five elements,
\[
A:=A_1=\{a\},\qquad
B:=B_1=\{b\},\qquad
Z:=H_0=\{z\},\qquad
X:=X_0,\qquad
Y:=Y_0.
\]
The five displayed classes are prime because each has a one-letter shortest representative.  Every other live non-unit class has one of the exact decompositions
\[
A_i=A^i,\qquad B_j=B^j,\qquad H_i=ZA^i,
\]
\[
X_d=
\begin{cases}
A^dX,&d>0,\\
X,&d=0,\\
XB^{-d},&d<0,
\end{cases}
\qquad
Y_d=
\begin{cases}
A^dY,&d>0,\\
Y,&d=0,\\
YB^{-d},&d<0,
\end{cases}
\]
and
\[
T_d=ZA^dX\qquad(d\ge0).
\]
Hence $q=5$; in particular the example has finite relative prime spectrum despite lying outside every finite $k,\ell$ window.

The normal forms also determine the valid branching rules.  A correct branching return to $X$ must have the form
\[
X\to A^mXB^m\qquad(m\ge1),
\]
and a correct branching return to $Y$ must have the form
\[
Y\to A^mYB^m\qquad(m\ge1).
\]
No other prime target admits a branching return.  When $m>1$, the right-hand side contains the proper prime-valued block $AXB$ or $AYB$, respectively.  Thus the only valid branching rules are
\[
X\to AXB,
\qquad
Y\to AYB.
\]
The canonical grammar is therefore
\[
\begin{aligned}
S&\to X\mid Y\mid ZX,\\
X&\to AXB\mid x,\\
Y&\to AYB\mid y,\\
A&\to a,\qquad B\to b,\qquad Z\to z.
\end{aligned}
\]
Here $q=5$ and $\tau=3$.  Every canonical prime word has length one, so $\CS_{\rm cut}=\varnothing$ and $\eta_{\rm cut}(L^\star,h)=0$.  Since $|\Sigma|=5$ and $|C_2\times C_2|=4$, Theorem~\ref{thm:ptld-strong} gives the explicit (not necessarily minimal) estimates
\[
|\CS^{\rm PTLD}|\le 5^2+5\cdot3+5+4=49
\]
and
\[
\max_{w\in\CS^{\rm PTLD}}|w|\le(5+1)3=18.
\]
All five primes are lexically anchored.  Hence
Corollary~\ref{cor:lexically-anchored-polynomial} applies: this same
nonregular witness lies in a PTLD subbranch with polynomial update time and
polynomial characteristic data in canonical grammar size.  Thus the
polynomial-data strong-learning branch itself reaches outside the entire
finite $k,\ell$-substitutable hierarchy, while retaining a five-prime
recursive canonical presentation.
\end{example}

\subsection{Finite-state strong reconstruction beyond FRP}

The direct learner above cannot cover the $36$-element witness, because that pair has infinitely many valid productions.  FSRP nevertheless supplies a canonical finite-state structural target through the minimal residual controllers of Section~\ref{sec:fsrp}.

\begin{theorem}[Weak-to-strong finite-state lifting]
\label{thm:fsrp-strong}
Fix $h:\Sigma^*\to M$.  Consider a class of $h$-substitutable context-free languages satisfying unit separation, finite relative prime spectrum, and FSRP.  Suppose the class has a weakly behaviorally correct positive-data learner whose hypotheses are CFGs.  Then there is a positive-data learner that identifies in the strong Gold sense the canonical finite-state relative-prime presentation, and hence its fixed canonical CFG compilation.
\end{theorem}

We first isolate the effective facts used by the lifting construction.

\begin{lemma}[Effective relative-class oracle from a correct CFG]
\label{lem:fsrp-class-oracle}
Let $G$ be a CFG with $L(G)=L$, where $L$ is $h$-substitutable.  Then factorhood and membership in every live relative class are decidable uniformly from $G$ and $h$.
\end{lemma}

\begin{proof}
For a word $u$, factorhood is equivalent to
\(L(G)\cap\Sigma^*u\Sigma^*\ne\varnothing.\)
The intersection of a CFL with a regular language is context-free and CFL emptiness is decidable, so factorhood is decidable.  If $u$ is live, enumerate pairs $(\ell,r)$ until $\ell ur\in L(G)$; such a pair exists and CFG membership is decidable.  Fix one such context.  Then, for every $v$,
\(v\in[u]_{L,h} \iff h(v)=h(u)\ \text{and}\ \ell vr\in L(G).\)
The forward implication follows from $u\equiv_Lv$.  Conversely, if the right-hand side holds, $u$ and $v$ have the same $h$-type and share the accepting context $(\ell,r)$, so $h$-substitutability gives $u\equiv_Lv$ and hence $u\theta_{L,h}v$.
\end{proof}

\begin{lemma}[Prime-spectrum stabilization]
\label{lem:fsrp-prime-stabilization}
Assume the relative prime spectrum is finite.  Using the oracle of Lemma~\ref{lem:fsrp-class-oracle}, there is a computable stagewise procedure whose finite prime hypotheses eventually stabilize to exactly $\Pcal_{L,h}$.
\end{lemma}

\begin{proof}
Enumerate live relative classes by their shortlex least representatives.  By Corollary~\ref{cor:finite-candidates}, every exact binary factorization of a live non-unit class $X$ occurs at one of the finitely many cuts of its shortest representative $x$.  For a cut $x=uv$, let $Y=[u]$ and $Z=[v]$.  Congruence gives $YZ\subseteq X$.  If equality fails, enumerate words of $X$ until a counterexample in $X\setminus YZ$ is found; membership in $YZ$ is decidable by testing the finitely many cuts of the candidate word with the relative-class oracle.  Thus every false decomposition candidate is eventually refuted.

At stage $s$, let $\Pcal_s$ consist of those classes among the first $s$ enumerated live non-unit classes for which every nontrivial cut candidate of the shortest representative has already been refuted by stage $s$; unresolved classes are simply omitted from $\Pcal_s$.  If $X$ is composite, at least one genuine exact cut is never refuted, so $X$ never enters $\Pcal_s$.  If $X$ is prime, all of its finitely many cut candidates are false and hence are eventually refuted, after which $X$ belongs permanently to every subsequent $\Pcal_s$.  Since the true relative prime spectrum is finite, all true primes are enumerated and certified by some common finite stage.  Therefore $\Pcal_s$ eventually stabilizes exactly to $\Pcal_{L,h}$.
\end{proof}

\begin{lemma}[Effective validity after prime stabilization]
\label{lem:fsrp-validity-oracle}
After the prime alphabet has stabilized, membership in each language $\Val_P$ is decidable uniformly.
\end{lemma}

\begin{proof}
For a prime sequence $\alpha=P_1\cdots P_k$, concatenate fixed canonical representatives of the $P_i$.  The relative-class oracle decides the quotient product of every contiguous block.  Hence it decides whether the whole sequence returns to $P$ and whether any proper contiguous block of length at least two returns to a prime.  These are exactly correctness and non-pleonasticity, so membership in $\Val_P$ is decidable.
\end{proof}

\begin{lemma}[Canonical DFA convergence]
\label{lem:fsrp-dfa-convergence}
Let $R\subseteq\Pcal^*$ be regular and suppose membership in $R$ is decidable.  There is a computable sequence of normalized DFAs that eventually stabilizes to the canonical minimal DFA of $R$.
\end{lemma}

\begin{proof}
At stage $t$, enumerate normalized complete DFAs with at most $t$ states and retain those agreeing with $R$ on all words of length at most $t$.  If no such DFA exists, output a fixed normalized one-state default DFA for that stage.  Otherwise choose first by minimum number of states and then by a fixed canonical encoding.  Let the minimal DFA of $R$ have $s$ states.  There are only finitely many normalized DFAs with at most $s$ states.  Every non-equivalent one differs from $R$ on a finite word; let $b$ be the maximum length of a shortest distinguishing word among these finitely many competitors.  For $t\ge\max\{s,b\}$ every smaller or same-size wrong DFA is eliminated, and the chosen DFA is exactly the canonical minimal DFA.  It remains so thereafter.
\end{proof}

\begin{lemma}[Stabilization of the lexical and start structure]
\label{lem:fsrp-start-stabilization}
Under unit separation and finite relative prime spectrum, the lexical component and the start component of $\mathfrak G^{\mathrm{fs}}_{L,h}$ are limit-computable from a correct CFG.
\end{lemma}

\begin{proof}
There are finitely many letters, so $\mathrm{Lex}(P)=P\cap\Sigma$ is determined by finitely many class-membership tests.  For each $m\in M$, the language $L\cap h^{-1}(m)$ is context-free and emptiness is decidable.  If nonempty, $h$-substitutability and the empty context imply that it is a single accepting relative class, so there are at most $|M|$ such classes.

For an accepting non-unit class $X$, choose a shortest representative $x$.  Corollary~\ref{cor:finite-candidates} gives a finite list of all exact prime-factorization candidates directly from the cuts of $x$.  For each candidate $\alpha$, congruence gives $\bar\alpha\subseteq X$ whenever its quotient product is $X$; if equality fails, enumerate words of $X$ until a counterexample in $X\setminus\bar\alpha$ appears.  Thus every false candidate is eventually eliminated and all exact start decompositions stabilize.  If $\eps\in L$, the empty start alternative is detected directly and recorded separately.
\end{proof}

\begin{proof}[Proof of Theorem~\ref{thm:fsrp-strong}]
Let $\mathcal W$ be the assumed weakly behaviorally correct learner and let $T$ be a positive text for a target $L$.  At stage $t$, apply the preceding canonicalization procedures with finite search budget $t$ to the CFG $G_t=\mathcal W(T[t])$ and output the resulting finite-state relative-prime presentation.

By weak behavioral correctness there is $t_0$ such that $L(G_t)=L$ for every $t\ge t_0$.  From that point onward Lemma~\ref{lem:fsrp-class-oracle} supplies the exact relative-class predicate, and Lemma~\ref{lem:fsrp-prime-stabilization} makes the finite prime alphabet stabilize.  Lemma~\ref{lem:fsrp-validity-oracle} then gives exact membership predicates for all $\Val_P$.  Since the target has FSRP, each $\Val_P$ is regular, so Lemma~\ref{lem:fsrp-dfa-convergence} makes every canonical minimal DFA stabilize.  Equality of residual languages represented by the stabilized DFAs is decidable by standard DFA equivalence; identifying equal residuals across the finitely many start languages $\Val_P$, and then discarding any occurrence of the empty residual $\varnothing$, therefore yields exactly the shared nonempty partial residual-state controller $\Qcal_{L,h}$ of Definition~\ref{def:canonical-fsrp}.  Lemma~\ref{lem:fsrp-start-stabilization} simultaneously stabilizes the lexical and start components.

All choices are canonical---shortlex representatives, normalized DFA encodings, and residual languages themselves as controller states---and depend only on $(L,h)$ once the weak grammar is language-correct.  Hence there is a stage after which the whole tuple $\mathfrak G^{\mathrm{fs}}_{L,h}$ is output literally unchanged.  The fixed compilation of Section~\ref{sec:fsrp} therefore also stabilizes to its canonical finite CFG.  This is strong Gold identification.
\end{proof}

\begin{remark}
Theorem~\ref{thm:fsrp-strong} is a computable limit result, not a polynomial-time claim.  The state bound of the target residual controller is not assumed known, and the DFA search is deliberately enumeration-based.  Its contribution is the canonical relative structural target being reconstructed, not the standard finite-automaton enumeration argument.  Theorem~\ref{thm:ptld-strong} is the efficient algebraically rigid subcase; the $36$-element UF--FRP witness lies outside that direct branch but inside the finite-state branch.
\end{remark}

\section{Current hierarchy and separations}

Throughout the following diagram, $L$ is assumed $h$-substitutable and unit-separated.  Implications involving FRP and the quadratic rule bound additionally assume a finite relative prime spectrum.

The established implication structure is now
\[
\begin{gathered}
\text{factor cancellation}
\Longrightarrow
\mathrm{PTLD}
\Longrightarrow
\mathrm{UF},
\\[1mm]
\mathrm{PTLD}
\Longrightarrow
\mathrm{VTE}
\Longrightarrow
\mathrm{VTE}_1
\Longrightarrow
\mathrm{FRP}
\Longrightarrow
\mathrm{FSRP},
\\[1mm]
\mathrm{PTLD}
\Longrightarrow
\mathrm{VTD}.
\end{gathered}
\]
Thus VTE and VTD are independent consequences of PTLD; no implication between VTE and VTD is asserted here.
The structural arrows out of PTLD are Theorems~\ref{thm:ptld-uf} and~\ref{thm:ptld-vte-vtd}.  The important non-implication $\mathrm{UF}\not\Rightarrow\mathrm{FRP}$ is Corollary~\ref{cor:uf-frp-separation}; Example~\ref{ex:nonuf-frp} gives the reverse non-implication, so UF and FRP are incomparable.

Theorems~\ref{thm:ptld-uf} and~\ref{thm:ptld-vte-vtd} give
\[
  \mathrm{PTLD}\Longrightarrow \mathrm{UF}+\mathrm{VTE}+\mathrm{VTD}+\mathrm{FRP},
  \qquad |\Vcal|\le q^2.
\]
The two nonregular direct-presentation examples clarify that PTLD is sufficient but not necessary:
\begin{center}
\begin{tabular}{lccccc}
\toprule
 & nonregular CFL & $h$-substitutable & PTLD & UF & FRP\\
\midrule
$L_0$ & yes & yes & no & yes & yes\\
$L^\star$ & yes & yes & yes & yes & yes\\
\bottomrule
\end{tabular}
\end{center}

\begin{center}
\scriptsize
\begin{tabular}{lccccc}
\toprule
Example & finite quotient & PTLD & exact UF & VTE$_1$ & presentation\\
\midrule
$L_0=\{a^nb^n\}^*$ & no & no & yes & yes & FRP, $\rho=3$\\
$L^\star$ & no & yes & yes & yes & FRP, $\rho=3$\\
finite FRP/non-UF witness & finite & no & no & not needed & FRP\\
$36$-element quotient witness & yes & no & yes & no & FSRP, not FRP\\
$L^{\mathrm{wrap}}$ & no & no & yes & no & nonregular CFL; FSRP, not FRP\\
\bottomrule
\end{tabular}
\end{center}

The profile $\Pi(L,h)=(q,r_{\rm res},m_{\rm ex};\rho_\partial,\kappa_\partial)$ separates prime spectrum, residual splitting, exact-factor multiplicity, and under-saturation complexity.  Section~\ref{sec:36} realizes the sharp combination $m_{\rm ex}=1$ with $\rho_\partial=\infty$, showing that exact-factor ambiguity and saturation defect are independent axes.

The learning hierarchy mirrors the presentation hierarchy but is not identical to it.  On the PTLD branch, Theorem~\ref{thm:ptld-strong} gives direct canonical SGOLD reconstruction with polynomial-time updates and an explicit finite characteristic sample; Lemma~\ref{lem:cut-prefix-escape} identifies its remaining quantitative parameter with prime-prefix escape.  The lexically anchored PTLD subbranch has $\eta_{\rm cut}=0$ and therefore polynomial characteristic data in explicit canonical grammar size, while Proposition~\ref{prop:compact-thickness-obstruction} shows that no such bound can be transferred uniformly from arbitrary compact CFG size.  On the broader FSRP branch, Theorem~\ref{thm:fsrp-strong} gives computable strong reconstruction of the minimal finite-state controller whenever a weakly behaviorally correct CFG learner is available.  The $36$-element witness and the infinite-quotient wrapped witness lie in the second branch but not the first, while Example~\ref{ex:Lstar} lies in the polynomial-data PTLD subbranch and is outside every finite $k,\ell$-substitutable class.

\section{Open problems and next steps}

\begin{question}[Absolute minimization of the UF--FRP separation]
The $36$-element witness of Section~\ref{sec:36} is quotient-minimal within its own monoid: every proper quotient has FRP.  Is there nevertheless a different finite monoid of size below $36$ admitting a unit-separated $h$-substitutable pair with exact UF but without FRP?  The present computation does not establish an absolute lower bound.
\end{question}

\begin{question}[Existence of non-FSRP finite-prime examples]
Does there exist a context-free fixed-$h$ pair $(L,h)$ that is $h$-substitutable and unit-separated, has finite relative prime spectrum, but fails FSRP?  Equivalently, can some canonical valid right-hand-side language $\Val_P$ be nonregular under these hypotheses?  We are not aware of such an example.  The wrapped witness of Section~\ref{sec:wrapped} shows that FSRP can be nontrivial even for a nonregular language with infinite relative quotient, but it still satisfies FSRP.  Thus the present results do not rule out the possibility that FSRP holds automatically for all context-free fixed-$h$ pairs with finite relative prime spectrum.
\end{question}

\begin{question}[FSRP versus context-free controller presentation]
If non-FSRP finite-prime examples exist, can one find one for which every $\Val_P$ is context-free but some $\Val_P$ is nonregular?  Such an example would separate finite-state structural presentation from finite context-free-controller presentation.  Proposition~\ref{prop:corr-cfl} shows that the corresponding correct-return languages $\Corr_P$ are already context-free; the extra difficulty lies entirely in factor-minimalization.
\end{question}

\begin{question}[Efficient learning beyond PTLD]
Theorem~\ref{thm:fsrp-strong} learns the FSRP residual controller only by a general limit enumeration of finite automata.  Which semantic conditions, weaker than PTLD but stronger than bare FSRP, provide an effective a priori bound or a polynomial procedure for recovering the minimal controller?  In particular, can bounded residual splitting, bounded exact-factor multiplicity, or a bounded defect-state parameter replace prime-target thinness in an efficient strong learner?
\end{question}

\begin{question}[Polynomial prime-prefix escape]
Lemma~\ref{lem:cut-prefix-escape} shows that under PTLD the remaining
cut-separation parameter is exactly a shortest prime-prefix escape length.
Can $\eta_{\rm cut}(L,h)$ be bounded by a polynomial in $q$ and $\tau$?
Equivalently, do the finitely many false prefix classes exposed by cuts of
canonical prime words always admit polynomial-length escape witnesses?
Corollary~\ref{cor:lexically-anchored-polynomial} answers this positively
with $\eta_{\rm cut}=0$ on the lexically anchored branch.
\end{question}

\begin{question}[Representation-size transfer under restricted presentations]
Proposition~\ref{prop:compact-thickness-obstruction} rules out any general
polynomial bound on canonical thickness in the size of an arbitrary compact
CFG, even for finite unary PTLD languages.  For which familiar restricted
CFG subclasses can $q$, $\tau$, and $\eta_{\rm cut}(L,h)$ nevertheless be
bounded polynomially in the size of the given generating grammar?  A
positive answer for a natural linear fixed-$h$ subclass would connect the
general weak polynomial bounds of the fixed-$h$ reconstruction paper
directly to polynomial strong structural learning.
\end{question}

\begin{question}[Multidimensional extension]
What is the correct analogue of exact relative prime factorization, semantic residual splitting, and under-saturation controllers for multidimensional syntactic congruence and bounded-fan-out MCFGs?  Yoshinaka and Clark's tuple congruence and non-permuting linear regular functions provide the natural compositional substrate~\cite{YoshinakaClark2012}.
\end{question}

\section{Conclusion}

The central structural separation is
\[
  \mathrm{UF}\not\Rightarrow\mathrm{FRP}
  \qquad\text{and}\qquad
  \mathrm{FRP}\subsetneq\mathrm{FSRP}.
\]
The $36$-element quotient witness has globally unique exact factorization but infinitely many valid prime returns, while the wrapped language $L^{\mathrm{wrap}}$ lifts the same phenomenon to a nonregular context-free language with infinite relative quotient and finite prime spectrum.  Together they show that exact-factor ambiguity, quotient infinitude, and direct-rule infinitude are genuinely different axes; FSRP captures the finite-state presentation level that remains when a finite direct rule list fails.

PTLD isolates a rigid branch of this theory.  It forces unique exact factorization, valid-tail exactness and determinism, and the quadratic rule bound.  The five-prime language $L^\star$ shows that this branch reaches beyond every finite $k,\ell$-substitutable class.  On the learning side, Relative-ASGOLD gives canonical strong reconstruction with polynomial-time updates and an explicit finite characteristic sample; its remaining general data-length parameter is the prime-prefix escape radius, which vanishes on the lexically anchored branch containing $L^\star$.  The unary doubling family shows separately that arbitrary compact CFGs may hide exponential canonical thickness, so representation-size transfer requires additional restrictions.

The main unresolved structural boundary lies outside FSRP: we are not aware of a context-free fixed-$h$ finite-prime example with a nonregular valid-return language.  It also remains open whether PTLD prime-prefix escape is polynomially bounded in canonical parameters.  Natural multidimensional and bounded-fan-out extensions remain to be developed.

\section*{Code and data availability}
The computational artifacts supporting the $36$-element witness are publicly available in a fixed GitHub snapshot at commit
\texttt{e38a9b09fbba7933777d7a00d7de4a7d77c4a916}:
\begin{quote}\small
\url{https://github.com/growupkuriyama-hub/lean_cfg_project/tree/e38a9b09fbba7933777d7a00d7de4a7d77c4a916/LeanCfgProject/Relative-Factorization}
\end{quote}
The snapshot contains the standard-library-only verifier \texttt{verify\_36.py} and the deterministic machine-readable certificate \texttt{certificate\_36.json}.  Appendix~\ref{app:36-computation} describes the theorem-facing checks.

\appendix
\section{Reproducible verification of the 36-element UF--FRP witness}
\label{app:36-computation}

The fixed snapshot cited in the Code and data availability statement contains the theorem-facing verifier \texttt{verify\_36.py} and the corresponding deterministic machine-readable certificate \texttt{certificate\_36.json}.

The file \texttt{verify\_\allowbreak 36.py} is a standard-library-only verifier for all theorem-facing finite claims in Section~\ref{sec:36}.  From the directory containing the two files, running
\begin{verbatim}
python3 verify_36.py --json certificate_36.generated.json
\end{verbatim}
reconstructs the $64$-element transformation monoid, computes the least monoid congruence $\sim_0$ generated by
\(h_0(ab)\sim_0 h_0(abaaab),\)
and then performs the following independent finite checks.

\begin{enumerate}[leftmargin=2.2em]
\item The quotient has $36$ elements and the identity congruence block
      contains only the identity transformation, hence
      $h^{-1}(1)=\{\eps\}$.
\item All $35$ nonidentity quotient classes are reachable by nonempty
      words, and every quotient element has a shortest representative
      of length at most eight.
\item Exact binary fiber-product equality identifies precisely the fifteen relative primes listed in Theorem~\ref{thm:36-finite-structure}.  The equality test is a finite product search comparing the DFA for the target positive fiber with the NFA for the concatenation of the positive factor fibers; a mismatch state is exactly a witness in their symmetric difference.
\item Lemma~\ref{lem:shortest-word-cut} reduces exact factorization to
      cuts of shortest representatives.  There are $222$ prime-labeled
      candidates.  The regular-language equality checker rejects $187$
      candidates and stores a shortest counterexample for each; exactly
      one candidate survives for each of the $35$ live non-unit classes.
\item The quotient multiplication satisfies $q^2=q^3$ for $q=h(baaa)$ and, with $\alpha=h(a)$ and $\beta=h(b)$, also verifies $\alpha q=\alpha q^2$ and $q\beta=q^2\beta$, together with
\(h(a)h(b)=h(a)qh(b)=h(a)q^2h(b)=h(ab).\)
      The proper interval classes occurring in
      $P_aP_q^mP_b$ are exactly
      $[abaaa]$, $[baaab]$, and $[baaabaaa]$, and the verifier checks
      the three exact composite identities displayed in Section~\ref{sec:36}.
\item The tail class is $H=[baaab]$, its unique exact factorization is
      $[b][aaab]$, and the setwise product $\overline{P_q^mP_b}$ is a strict subset of
      $H$ for every $m\ge1$.  The certificate records explicit
      counterexamples for the periodic cases.
\item The prime-avoidance automaton has a productive lasso.  After the
      prefix $P_aP_q^2$ it reaches the state whose whole-prefix class
      is $[abaaa]$; its proper-suffix classes are $[baaa]$ and
      $[baaabaaa]$.  Another $[baaa]$ is a self-loop, while
      the final symbol $[b]$ returns to the prime $[ab]$.
\item\label{check:quotient-minimality} Exhaustive congruence closure yields exactly twenty monoid
      congruences on $M_{36}$.  The nineteen proper quotient sizes are
\(11,9,8,7,7,6,6,5,5,5,4,4,4,3,3,2,2,2,1.\)
      For each induced observer quotient the verifier recomputes the
      live relative classes and relative primes and checks that the
      productive part of the prime-avoidance graph is acyclic.  Hence
      every proper quotient has FRP.  When a proper observer quotient
      maps nonempty words to its identity, the verifier keeps the
      syntactic unit $\{\eps\}$ separate from the positive identity
      fiber, as required by $\theta=\equiv_L\cap\ker h$.
\end{enumerate}

The committed \texttt{certificate\_\allowbreak 36.json} contains the $35$ unique exact factorizations, all $187$ rejected shortest-cut candidates and their counterexamples, the three composite interval certificates, the prime-avoidance lasso data, and the audit record for every proper congruence quotient.  The generated certificate is deterministic and can therefore be compared byte-for-byte with the committed \texttt{certificate\_\allowbreak 36.json} in the fixed repository snapshot, for example by running
\begin{verbatim}
cmp certificate_36.generated.json certificate_36.json
\end{verbatim}
after the verifier completes.

\end{document}